\documentclass[english,12pt]{article}
\usepackage[T1]{fontenc}
\usepackage[latin9]{inputenc}
\usepackage[a4paper]{geometry}
\usepackage{color}
\usepackage{babel}
\usepackage{float}
\usepackage{url}
\usepackage{bm}
\usepackage{amsmath}
\usepackage{amsthm}
\usepackage{amssymb}
\usepackage{graphicx}
\usepackage{esint}
\usepackage[authoryear]{natbib}
\usepackage[unicode=true,pdfusetitle,
 bookmarks=true,bookmarksnumbered=false,bookmarksopen=false,
 breaklinks=false,pdfborder={0 0 0},backref=false,colorlinks=true]
 {hyperref}
\hypersetup{
 citecolor=blue,linkcolor=blue}

\makeatletter

\providecommand{\tabularnewline}{\\}
\newcommand{\lyxdot}{.}

\floatstyle{ruled}
\newfloat{algorithm}{tbp}{loa}
\providecommand{\algorithmname}{Algorithm}
\floatname{algorithm}{\protect\algorithmname}

\theoremstyle{plain}
\newtheorem{thm}{\protect\theoremname}
  \theoremstyle{plain}
  \newtheorem{prop}[thm]{\protect\propositionname}

\usepackage[gennarrow]{eurosym}

\@ifundefined{showcaptionsetup}{}{%
 \PassOptionsToPackage{caption=false}{subfig}}
\usepackage{subfig}
\makeatother

  \providecommand{\propositionname}{Proposition}
\providecommand{\theoremname}{Theorem}

\begin{document}

\title{The iterated auxiliary particle filter}

\author{Pieralberto Guarniero$^{*}$, Adam M. Johansen$^{*}$ and Anthony
Lee$^{*,\dagger}$\\
$^{*}$Department of Statistics, University of Warwick\\
$^{\dagger}$Alan Turing Institute}
\maketitle
\begin{abstract}
We present an offline, iterated particle filter to facilitate statistical
inference in general state space hidden Markov models. Given a model
and a sequence of observations, the associated marginal likelihood
$L$ is central to likelihood-based inference for unknown statistical
parameters. We define a class of ``twisted'' models: each member
is specified by a sequence of positive functions $\bm{\psi}$ and
has an associated $\bm{\psi}$-auxiliary particle filter that provides
unbiased estimates of $L$. We identify a sequence $\bm{\psi}^{*}$
that is optimal in the sense that the $\bm{\psi}^{*}$-auxiliary particle
filter's estimate of $L$ has zero variance. In practical applications,
$\bm{\psi}^{*}$ is unknown so the $\bm{\psi}^{*}$-auxiliary particle
filter cannot straightforwardly be implemented. We use an iterative
scheme to approximate $\bm{\psi}^{*}$, and demonstrate empirically
that the resulting iterated auxiliary particle filter significantly
outperforms the bootstrap particle filter in challenging settings.
Applications include parameter estimation using a particle Markov
chain Monte Carlo algorithm.
\end{abstract}

\newcommand{\sectionpsi}[0]{\texorpdfstring{$\bm{\psi}$}{psi}}
\newcommand{\sectionpsistar}[0]{\texorpdfstring{$\bm{\psi^*}$}{psi*}}

\section{Introduction}

Particle filtering, or sequential Monte Carlo (SMC), methodology involves
the simulation over time of an artificial particle system $(\xi_{t}^{i};\:t\in\left\{ 1,\dots,T\right\} ,\:i\in\{1,\ldots,N\})$.
It is particularly suited to numerical approximation of integrals
of the form
\begin{equation}
Z:=\int_{\mathsf{X}^{T}}\mu_{1}\left(x_{1}\right)g_{1}\left(x_{1}\right)\prod_{t=2}^{T}f_{t}\left(x_{t-1},x_{t}\right)g_{t}\left(x_{t}\right)dx_{1:T},\label{eq:Z}
\end{equation}
where $\mathsf{X}=\mathbb{R}^{d}$ for some $d\in\mathbb{N}$, $T\in\mathbb{N}$,
$x_{1:T}:=(x_{1},\ldots,x_{T})$, $\mu_{1}$ is a probability density
function on $\mathsf{X}$, each $f_{t}$ a transition density on $\mathsf{X}$,
and each $g_{t}$ is a bounded, continuous and non-negative function.
Algorithm~\ref{alg:Particle-Filter} describes a particle filter,
using which an estimate of (\ref{eq:Z}) can be computed as
\begin{equation}
Z^{N}:=\prod_{t=1}^{T}\left[\frac{1}{N}\sum_{i=1}^{N}g_{t}(\xi_{t}^{i})\right].\label{eq:ZN}
\end{equation}
\begin{algorithm}[h]
\caption{A Particle Filter\label{alg:Particle-Filter}}

\begin{enumerate}
\item Sample $\xi_{1}^{i}\sim\mu_{1}$ independently for $i\in\{1,\ldots,N\}$.
\item For $t=2,\ldots,T$, sample independently
\[
\xi_{t}^{i}\sim\frac{\sum_{j=1}^{N}g_{t-1}(\xi_{t-1}^{j})f_{t}(\xi_{t-1}^{j},\cdot)}{\sum_{j=1}^{N}g_{t-1}(\xi_{t-1}^{j})},\qquad i\in\{1,\ldots,N\}.
\]
\end{enumerate}
\end{algorithm}

Particle filters were originally applied to statistical inference
for hidden Markov models (HMMs) by \citet{Gordon1993}, and this setting
remains an important application. Letting $\mathsf{Y}=\mathbb{R}^{d'}$
for some $d'\in\mathbb{N}$, an HMM is a Markov chain evolving on
$\mathsf{X}\times\mathsf{Y}$, $(X_{t},Y_{t})_{t\in\mathbb{N}}$,
where $(X_{t})_{t\in\mathbb{N}}$ is itself a Markov chain and for
$t\in\{1,\ldots,T\}$, each $Y_{t}$ is conditionally independent
of all other random variables given $X_{t}$. In a time-homogeneous
HMM, letting $\mathbb{P}$ denote the law of this bivariate Markov
chain, we have
\begin{equation}
\mathbb{P}\left(X_{1:T}\in A,Y_{1:T}\in B\right):=\int_{A\times B}\mu\left(x_{1}\right)g\left(x_{1},y_{1}\right)\prod_{t=2}^{T}f\left(x_{t-1},x_{t}\right)g\left(x_{t},y_{t}\right)dx_{1:T}dy_{1:T},\label{eq:P}
\end{equation}
where $\mu:\mathsf{X}\rightarrow\mathbb{R}_{+}$ is a probability
density function, $f:\mathsf{X}\times\mathsf{X}\rightarrow\mathbb{R}_{+}$
a transition density, $g:\mathsf{X}\times\mathsf{Y}\rightarrow\mathbb{R}_{+}$
an observation density and $A$ and $B$ measurable subsets of $\mathsf{X}^{T}$
and $\mathsf{Y}^{T}$, respectively. Statistical inference is often
conducted upon the basis of a realization $y_{1:T}$ of $Y_{1:T}$
for some finite $T$, which we will consider to be fixed throughout
the remainder of the paper. Letting $\mathbb{E}$ denote expectations
w.r.t. $\mathbb{P}$, our main statistical quantity of interest is
$L:=\mathbb{E}\left[\prod_{t=1}^{T}g\left(X_{t},y_{t}\right)\right]$,
the marginal likelihood associated with $y_{1:T}$. In the above,
we take $\mathbb{R}_{+}$ to be the non-negative real numbers, and
assume throughout that $L>0$.

Running Algorithm~\ref{alg:Particle-Filter} with
\begin{equation}
\mu_{1}=\mu,\qquad f_{t}=f,\qquad g_{t}(x)=g(x,y_{t}),\label{eq:bootstrap_defn}
\end{equation}
corresponds exactly to running the bootstrap particle filter (BPF)
of \citet{Gordon1993}, and we observe that when (\ref{eq:bootstrap_defn})
holds, the quantity $Z$ defined in (\ref{eq:Z}) is identical to
$L$, so that $Z^{N}$ defined in (\ref{eq:ZN}) is an approximation
of $L$. In applications where $L$ is the primary quantity of interest,
there is typically an unknown statistical parameter $\theta\in\Theta$
that governs $\mu$, $f$ and $g$, and in this setting the map $\theta\mapsto L(\theta)$
is the likelihood function. We continue to suppress the dependence
on $\theta$ from the notation until Section~\ref{sec:Applications-and}.

The accuracy of the approximation $Z^{N}$ has been studied extensively.
For example, the expectation of $Z^{N}$, under the law of the particle
filter, is exactly $Z$ for any $N\in\mathbb{N}$, and $Z^{N}$ converges
almost surely to $Z$ as $N\rightarrow\infty$; these can be seen
as consequences of \citet[Theorem~7.4.2]{DelMoral2004}. For practical
values of $N$, however, the quality of the approximation can vary
considerably depending on the model and/or observation sequence. When
used to facilitate parameter estimation using, e.g., particle Markov
chain Monte Carlo \citep{Andrieu2010}, it is desirable that the accuracy
of $Z^{N}$ be robust to small changes in the model and this is not
typically the case.

In Section~\ref{sec:Twisted-models-and} we introduce a family of
``twisted HMMs'', parametrized by a sequence of positive functions
$\bm{\psi}:=(\psi_{1},\ldots,\psi_{T})$. Running a particle filter
associated with any of these twisted HMMs provides unbiased and strongly
consistent estimates of $L$. Some specific definitions of $\bm{\psi}$
correspond to well-known modifications of the BPF, and the algorithm
itself can be viewed as a generalization of the auxiliary particle
filter (APF) of \citet{Pitt1999}. Of particular interest is a sequence
$\bm{\psi}^{*}$ for which $Z^{N}=L$ with probability $1$. In general,
$\bm{\psi}^{*}$ is not known and the corresponding APF cannot be
implemented, so our main focus in Section~\ref{sec:Approximation-of-}
is approximating the sequence $\bm{\psi}^{*}$ iteratively, and defining
final estimates through use of a simple stopping rule. In the applications
of Section~\ref{sec:Applications-and} we find that the resulting
estimates significantly outperform the BPF, and exhibit some robustness
to both increases in the dimension of the latent state space $\mathsf{X}$
and changes in the model parameters. There are some restrictions on
the class of transition densities and the functions $\psi_{1},\ldots,\psi_{T}$
that can be used in practice, which we discuss.

This work builds upon a number of methodological advances, most notably
the twisted particle filter \citep{whiteley2014twisted}, the APF
\citep{Pitt1999}, block sampling \citep{Doucet2006}, and look-ahead
schemes \citep{Lin2013}. In particular, the sequence $\bm{\psi}^{*}$
is closely related to the generalized eigenfunctions described in
\citet{whiteley2014twisted}, but in that work the particle filter
as opposed to the HMM was twisted to define alternative approximations
of $L$. For simplicity, we have presented the BPF in which multinomial
resampling occurs at each time step. Commonly employed modifications
of this algorithm include adaptive resampling \citep{Kong1994,Liu1995}
and alternative resampling schemes \citep[see, e.g.,][]{Doucet2005}.
Generalization to the time-inhomogeneous HMM setting is fairly straightforward,
so we restrict ourselves to the time-homogeneous setting for clarity
of exposition.

\section{Twisted models and the \sectionpsi-auxiliary particle filter\label{sec:Twisted-models-and}}

\noindent \begin{flushleft}
Given an HMM $(\mu,f,g)$ and a sequence of observations $y_{1:T}$,
we introduce a family of alternative twisted models based on a sequence
of real-valued, bounded, continuous and positive functions\emph{ $\boldsymbol{\psi}:=\left(\psi_{1},\psi_{2},\dots,\psi_{T}\right)$}.
Letting, for an arbitrary transition density $f$ and function $\psi$,
$f(x,\psi):=\int_{\mathsf{X}}f\left(x,x^{\prime}\right)\psi\left(x^{\prime}\right)dx'$,
we define a sequence of normalizing functions\emph{ $(\tilde{\psi}_{1},\tilde{\psi}_{2},\dots,\tilde{\psi}_{T})$}
on $\mathsf{X}$ by $\tilde{\psi}_{t}(x_{t}):=f\left(x_{t},\psi_{t+1}\right)$
for $t\in\{1,\ldots,T-1\}$, $\tilde{\psi}_{T}\equiv1$, and a normalizing
constant $\tilde{\psi}_{0}:=\int_{\mathsf{X}}\mu\left(x_{1}\right)\psi_{1}\left(x_{1}\right)dx_{1}$.
We then define the\emph{ }twisted model via the following sequence
of twisted initial and transition densities\emph{
\begin{equation}
\mu_{1}^{\boldsymbol{\psi}}(x_{1}):=\frac{\mu(x_{1})\psi_{1}(x_{1})}{\tilde{\psi}_{0}},\qquad f_{t}^{\boldsymbol{\psi}}(x_{t-1},x_{t}):=\frac{f\left(x_{t-1},x_{t}\right)\psi_{t}\left(x_{t}\right)}{\tilde{\psi}_{t-1}\left(x_{t-1}\right)},\quad t\in\{2,\ldots,T\},\label{eq:mufpsi}
\end{equation}
}and the sequence of positive functions
\begin{equation}
g_{1}^{\boldsymbol{\psi}}\left(x_{1}\right):=g\left(x_{1},y_{1}\right)\frac{\tilde{\psi}_{1}\left(x_{1}\right)}{\psi_{1}\left(x_{1}\right)}\tilde{\psi}_{0},\qquad g_{t}^{\boldsymbol{\psi}}\left(x_{t}\right):=g\left(x_{t},y_{t}\right)\frac{\tilde{\psi}_{t}\left(x_{t}\right)}{\psi_{t}\left(x_{t}\right)},\quad t\in\{2,\ldots T\},\label{eq:gpsi}
\end{equation}
which play the role of observation densities in the twisted model.
Our interest in this family is motivated by the following invariance
result.
\par\end{flushleft}
\begin{prop}
\label{prop:notunique}If $\boldsymbol{\psi}$ is a sequence of bounded,
continuous and positive functions, and
\[
Z_{\bm{\psi}}:=\int_{\mathsf{X}^{T}}\mu_{1}^{\bm{\psi}}\left(x_{1}\right)g_{1}^{\bm{\psi}}\left(x_{1}\right)\prod_{t=2}^{T}f_{t}^{\bm{\psi}}\left(x_{t-1},x_{t}\right)g_{t}^{\bm{\psi}}\left(x_{t}\right)dx_{1:T},
\]
then $Z_{\bm{\psi}}=L$.\end{prop}
\begin{proof}
We observe that
\begin{eqnarray*}
 &  & \mu_{1}^{\boldsymbol{\psi}}\left(x_{1}\right)g_{1}^{\boldsymbol{\psi}}\left(x_{1}\right)\prod_{t=2}^{T}f_{t}^{\boldsymbol{\psi}}\left(x_{t-1},x_{t}\right)g_{t}^{\boldsymbol{\psi}}\left(x_{t}\right)\\
 & = & \frac{\mu(x_{1})\psi_{1}(x_{1})}{\tilde{\psi}_{0}}g_{1}\left(x_{1}\right)\frac{\tilde{\psi}_{1}\left(x_{1}\right)}{\psi_{1}\left(x_{1}\right)}\tilde{\psi}_{0}\cdot\prod_{t=2}^{T}\frac{f\left(x_{t-1},x_{t}\right)\psi_{t}\left(x_{t}\right)}{\tilde{\psi}_{t-1}\left(x_{t-1}\right)}g_{t}\left(x_{t}\right)\frac{\tilde{\psi}_{t}\left(x_{t}\right)}{\psi_{t}\left(x_{t}\right)}\\
 & = & \mu\left(x_{1}\right)g_{1}\left(x_{1}\right)\prod_{t=2}^{T}f\left(x_{t-1},x_{t}\right)g_{t}\left(x_{t}\right),
\end{eqnarray*}
and the result follows.
\end{proof}
From a methodological perspective, Proposition~\ref{prop:notunique}
makes clear a particular sense in which the L.H.S. of (\ref{eq:Z})
is common to an entire family of $\mu_{1}$, $(f_{t})_{t\in\{2,\ldots,T\}}$
and $(g_{t})_{t\in\{1,\ldots,T\}}$. The BPF associated with the twisted
model corresponds to choosing
\begin{equation}
\mu_{1}=\mu^{\bm{\psi}},\qquad f_{t}=f_{t}^{\bm{\psi}},\qquad g_{t}=g_{t}^{\bm{\psi}},\label{eq:psi_boot_defn}
\end{equation}
in Algorithm~\ref{alg:Particle-Filter}; to emphasize the dependence
on $\bm{\psi}$, we provide in Algorithm~\ref{alg:apf_psi} the corresponding
algorithm and we will denote approximations of $L$ by $Z_{\bm{\psi}}^{N}$.
We demonstrate below that the BPF associated with the twisted model
can also be viewed as an APF associated with the sequence $\bm{\psi}$,
and so refer to this algorithm as the $\bm{\psi}$-APF. Since the
class of $\bm{\psi}$-APF's is very large, it is natural to consider
whether there is an optimal choice of $\bm{\psi}$, in terms of the
accuracy of the approximation $Z_{\bm{\psi}}^{N}$: the following
Proposition describes such a sequence.
\begin{algorithm}[h]
\caption{$\bm{\psi}$-Auxiliary Particle Filter\label{alg:apf_psi}}

\begin{enumerate}
\item Sample $\xi_{1}^{i}\sim\mu^{\bm{\psi}}$ independently for $i\in\{1,\ldots,N\}$.
\item For $t=2,\ldots,T$, sample independently
\[
\xi_{t}^{i}\sim\frac{\sum_{j=1}^{N}g_{t-1}^{\bm{\psi}}(\xi_{t-1}^{j})f_{t}^{\bm{\psi}}(\xi_{t-1}^{j},\cdot)}{\sum_{j=1}^{N}g_{t-1}^{\bm{\psi}}(\xi_{t-1}^{j})},\qquad i\in\{1,\ldots,N\}.
\]
\end{enumerate}
\end{algorithm}

\begin{prop}
\label{prop:optimalseq}Let $\bm{\psi}^{*}:=(\psi_{1}^{*},\ldots,\psi_{T}^{*})$,
where $\psi_{T}^{*}\left(x_{T}\right):=g(x_{T},y_{T})$, and 
\begin{equation}
\psi_{t}^{*}\left(x_{t}\right):=g\left(x_{t},y_{t}\right)\mathbb{E}\left[\prod_{p=t+1}^{T}g\left(X_{p},y_{p}\right)\biggl|\left\{ X_{t}=x_{t}\right\} \right],\quad x_{t}\in\mathsf{X},\label{eq:psi*t}
\end{equation}
for $t\in\{1,\ldots,T-1\}$. Then, $Z_{\bm{\psi}^{*}}^{N}=L$ with
probability $1$.\end{prop}
\begin{proof}
It can be established that
\[
g(x_{t},y_{t})\tilde{\psi}_{t}^{*}(x_{t})=\psi_{t}^{*}(x_{t}),\qquad t\in\{1,\ldots,T\},\qquad x_{t}\in\mathsf{X},
\]
and so we obtain from (\ref{eq:gpsi}) that $g_{1}^{\boldsymbol{\psi}^{*}}\equiv\tilde{\psi}_{0}^{*}$
and $g_{t}^{\boldsymbol{\psi}^{*}}\equiv1$ for $t\in\{2,\ldots,T\}$.
Hence,
\[
Z_{N}^{\boldsymbol{\psi}^{*}}=\prod_{t=1}^{T}\left[\frac{1}{N}\sum_{i=1}^{N}g_{t}^{\boldsymbol{\psi}^{*}}\left(\xi_{t}^{i}\right)\right]=\tilde{\psi}_{0}^{*},
\]
with probability $1$. To conclude, we observe that
\begin{eqnarray*}
\tilde{\psi}_{0}^{*} & = & \int_{\mathsf{X}}\mu\left(x_{1}\right)\psi_{1}^{*}\left(x_{1}\right)dx_{1}=\int_{\mathsf{X}}\mu\left(x_{1}\right)\mathbb{E}\left[\prod_{t=1}^{T}g\left(X_{t},y_{t}\right)\biggl|\left\{ X_{1}=x_{1}\right\} \right]dx_{1}\\
 & = & \mathbb{E}\left[\prod_{t=1}^{T}g\left(X_{t},y_{t}\right)\right]=L.\qedhere
\end{eqnarray*}

\end{proof}
Implementation of Algorithm~\ref{alg:apf_psi} requires that one
can sample according to $\mu_{1}^{\bm{\psi}}$ and $f_{t}^{\bm{\psi}}(x,\cdot)$
and compute $g_{t}^{\boldsymbol{\psi}}$ pointwise. This imposes restrictions
on the choice of $\bm{\psi}$ in practice, since one must be able
to compute both $\psi_{t}$ and $\tilde{\psi}_{t}$ pointwise. In
general models, the sequence $\bm{\psi}^{*}$ cannot be used for this
reason as (\ref{eq:psi*t}) cannot be computed explicitly. However,
since Algorithm~\ref{alg:apf_psi} is valid for any sequence of positive
functions $\bm{\psi}$, we can interpret Proposition~\ref{prop:optimalseq}
as motivating the effective design of a particle filter by solving
a sequence of function approximation problems.

Alternatives to the BPF have been considered before (see, e.g., the
``locally optimal'' proposal in \citealt{Doucet2000} and the discussion
in \citealt[Section~2.4.2]{DelMoral2004}). The family of particle
filters we have defined using $\bm{\psi}$ are unusual, however, in
that $g_{t}^{\bm{\psi}}$ is a function only of $x_{t}$ rather than
$(x_{t-1},x_{t})$; other approaches in which the particles are sampled
according to a transition density that is not $f$ typically require
this extension of the domain of these functions. This is again a consequence
of the fact that the $\bm{\psi}$-APF can be viewed as a BPF for a
twisted model. This feature is shared by the fully adapted APF of
\citet{Pitt1999}, when recast as a standard particle filter for an
alternative model as in \citet{Johansen2008}, and which is obtained
as a special case of Algorithm~\ref{alg:apf_psi} when $\psi_{t}(\cdot)\equiv g(\cdot,y_{t})$
for each $t\in\{1,\ldots,T\}$. We view the approach here as generalizing
that algorithm for this reason.

It is possible to recover other existing methodological approaches
as BPFs for twisted models. In particular, when each element of $\bm{\psi}$
is a constant function, we recover the standard BPF of \citet{Gordon1993}.
Setting $\psi_{t}\left(x_{t}\right)=g\left(x_{t},y_{t}\right)$ gives
rise to the fully adapted APF. By taking, for some $k\in\mathbb{N}$
and each $t\in\{1,\ldots,T\}$,
\begin{align}
\psi_{t}\left(x_{t}\right) & =g\left(x_{t},y_{t}\right)\mathbb{E}\left[\prod_{p=t+1}^{(t+k)\wedge T}g\left(X_{p},y_{p}\right)\biggl|\left\{ X_{t}=x_{t}\right\} \right],\quad x_{t}\in\mathsf{X},\label{eq:psi_lookahead}
\end{align}
$\bm{\psi}$ corresponds to a sequence of look-ahead functions \citep[see, e.g.,][]{Lin2013}
and one can recover idealized versions of the delayed sample method
of \citet{chen2000adaptive} (see also the fixed-lag smoothing approach
in \citealt{clapp1999fixed}), and the block sampling particle filter
of \citet{Doucet2006}. When $k\geq T-1$, we obtain the sequence
$\bm{\psi}^{*}$. Just as $\bm{\psi}^{*}$ cannot typically be used
in practice, neither can the exact look-ahead strategies obtained
by using (\ref{eq:psi_lookahead}) for some fixed $k$. In such situations,
the proposed look-ahead particle filtering strategies are not $\bm{\psi}$-APFs,
and their relationship to the $\bm{\psi}^{*}$-APF is consequently
less clear. We note that the offline setting we consider here affords
us the freedom to define twisted models using the entire data record
$y_{1:T}$. The APF was originally introduced to incorporate a single
additional observation, and could therefore be implemented in an online
setting, i.e. the algorithm could run while the data record was being
produced.

\section{Function approximations and the iterated APF\label{sec:Approximation-of-}}

\subsection{Asymptotic variance of the \sectionpsi-APF}

Since it is not typically possible to use the sequence $\bm{\psi}^{*}$
in practice, we propose to use an approximation of each member of
$\bm{\psi}^{*}$. In order to motivate such an approximation, we provide
a Central Limit Theorem, adapted from a general result due to \citet[Chapter~9]{DelMoral2004}.
It is convenient to make use of the fact that the estimate $Z_{\bm{\psi}}^{N}$
is invariant to rescaling of the functions $\psi_{t}$ by constants,
and we adopt now a particular scaling that simplifies the expression
of the asymptotic variance. In particular, we let 
\[
\bar{\psi}_{t}(x):=\frac{\psi_{t}(x)}{\mathbb{E}\left[\psi_{t}\left(X_{t}\right)\mid\left\{ Y_{1:t-1}=y_{1:t-1}\right\} \right]},\qquad\bar{\psi}_{t}^{*}(x):=\frac{\psi_{t}^{*}(x)}{\mathbb{E}\left[\psi_{t}^{*}\left(X_{t}\right)\mid\left\{ Y_{1:t-1}=y_{1:t-1}\right\} \right]}.
\]

\begin{prop}
\label{prop:clt}Let $\bm{\psi}$ be a sequence of bounded, continuous
and positive functions. Then
\[
\sqrt{N}\left(\frac{Z_{\bm{\psi}}^{N}}{Z}-1\right)\overset{d}{\longrightarrow}\mathcal{N}(0,\sigma_{\bm{\psi}}^{2}),
\]
where, 
\begin{equation}
\sigma_{\bm{\psi}}^{2}:=\sum_{t=1}^{T}\left\{ \mathbb{E}\left[\frac{\bar{\psi}_{t}^{*}\left(X_{t}\right)}{\bar{\psi}_{t}\left(X_{t}\right)}\left|\vphantom{\frac{\psi_{t}^{*}\left(X_{t}\right)}{\psi_{t}\left(X_{t}\right)}}\right.\left\{ \vphantom{\frac{\psi_{t}^{*}\left(X_{t}\right)}{\psi_{t}\left(X_{t}\right)}}Y_{1:T}=y_{1:T}\right\} \right]-1\right\} .\label{eq:avarpsi}
\end{equation}

\end{prop}
We emphasize that Proposition~\ref{prop:clt}, whose proof can be
found in the Appendix, follows straightforwardly from existing results
for Algorithm~\ref{alg:Particle-Filter}, since the $\bm{\psi}$-APF
can be viewed as a BPF for the twisted model defined by $\bm{\psi}$.
For example, in the case $\bm{\psi}$ consists only of constant functions,
we obtain the standard asymptotic variance for the BPF
\[
\sigma^{2}=\sum_{t=1}^{T}\left\{ \mathbb{E}\left[\bar{\psi}_{t}^{*}\left(X_{t}\right)\mid\left\{ Y_{1:T}=y_{1:T}\right\} \right]-1\right\} .
\]
From Proposition~\ref{prop:clt} we can deduce that $\sigma_{\bm{\psi}}^{2}$
tends to $0$ as $\bm{\psi}$ approaches $\bm{\psi}^{*}$ in an appropriate
sense. Hence, Propositions~\ref{prop:optimalseq} and~\ref{prop:clt}
together provide some justification for designing particle filters
by approximating the sequence $\bm{\psi}^{*}$.

\subsection{Classes of $f$ and \sectionpsi\label{sub:Classes-of-}}

While the $\bm{\psi}$-APF described in Section~\ref{sec:Twisted-models-and}
and the asymptotic results just described are valid very generally,
practical implementation of the $\bm{\psi}$-APF does impose some
restrictions jointly on the transition densities $f$ and functions
in $\bm{\psi}$. Here we consider only the case where the HMM's initial
distribution is a mixture of Gaussians and $f$ is a member of $\mathcal{F}$,
the class of transition densities of the form 
\begin{equation}
f\left(x,\cdot\right)=\sum_{k=1}^{M}c_{k}(x)\mathcal{N}\left(\:\cdot\:;a_{k}\left(x\right),b_{k}\left(x\right)\right),\label{eq:Fclass}
\end{equation}
where $M\in\mathbb{N}$, and $(a_{k})_{k\in\{1,\ldots,M\}}$ and $(b_{k})_{k\in\{1,\ldots,M\}}$
are sequences of mean and covariance functions, respectively and $(c_{k})_{k\in\{1,\ldots,M\}}$
a sequence of $\mathbb{R}_{+}$-valued functions with $\sum_{k=1}^{M}c_{k}(x)=1$
for all $x\in\mathsf{X}$. Let $\Psi$ define the class of functions
of the form 
\begin{equation}
\psi(x)=C+\sum_{k=1}^{M}c_{k}\mathcal{N}\left(x;a_{k},b_{k}\right),\label{eq:basis}
\end{equation}
where $M\in\mathbb{N}$, $C\in\mathbb{R}_{+}$, and $(a_{k})_{k\in\{1,\ldots,M\}}$,
$(b_{k})_{k\in\{1,\ldots,M\}}$ and $(c_{k})_{k\in\{1,\ldots,M\}}$
are a sequence of means, covariances and positive real numbers, respectively.
When $f\in\mathcal{F}$ and each $\psi_{t}\in\Psi$, it is straightforward
to implement Algorithm~\ref{alg:apf_psi} since, for each $t\in\{1,\ldots,T\}$,
both $\psi_{t}(x)$ and $\tilde{\psi}_{t-1}(x)=f(x,\psi_{t})$ can
be computed explicitly and $f_{t}^{\bm{\psi}}(x,\cdot)$ is a mixture
of normal distributions whose component means and covariance matrices
can also be computed. Alternatives to this particular setting are
discussed in Section~\ref{sec:Discussion}.

\subsection{Recursive approximation of \sectionpsistar}

The ability to compute $f(\cdot,\psi_{t})$ pointwise when $f\in\mathcal{F}$
and $\psi_{t}\in\Psi$ is also instrumental in the recursive function
approximation scheme we now describe. Our approach is based on the
following observation.
\begin{prop}
\label{prop:backwards_recursion}The sequence $\boldsymbol{\psi}^{*}$
satisfies $\psi_{T}^{*}\left(x_{T}\right)=g\left(x_{T},y_{T}\right)$,
$x_{T}\in\mathsf{X}$ and
\begin{equation}
\psi_{t}^{*}\left(x_{t}\right)=g\left(x_{t},y_{t}\right)f\left(x_{t},\psi_{t+1}^{*}\right),\quad x_{t}\in\mathsf{X},\quad t\in\{1,\ldots,T-1\}.\label{eq:recursive_psistar}
\end{equation}
\end{prop}
\begin{proof}
The definition of $\bm{\psi}^{*}$ provides that $\psi_{T}^{*}\left(x_{T}\right)=g\left(x_{T},y_{T}\right)$.
For $t\in\{1,\ldots,T-1\}$,
\begin{eqnarray*}
 &  & g\left(x_{t},y_{t}\right)f\left(x_{t},\psi_{t+1}^{*}\right)\\
 & = & g\left(x_{t},y_{t}\right)\int_{\mathsf{X}}f\left(x_{t},x_{t+1}\right)\mathbb{E}\left[\prod_{p=t+1}^{T}g\left(X_{p},y_{p}\right)\mid\left\{ X_{t+1}=x_{t+1}\right\} \right]dx_{t+1}\\
 & = & g\left(x_{t},y_{t}\right)\mathbb{E}\left[\prod_{p=t+1}^{T}g\left(X_{p},y_{p}\right)\mid\left\{ X_{t}=x_{t}\right\} \right]\\
 & = & \psi_{t}^{*}\left(x_{t}\right).\qedhere
\end{eqnarray*}

\end{proof}
Let $(\xi_{1}^{1:N},\ldots,\xi_{T}^{1:N})$ be random variables obtained
by running a particle filter. We propose to approximate $\bm{\psi}^{*}$
by Algorithm~\ref{alg:Function-approximations}, for which we define
$\psi_{T+1}\equiv1$. This algorithm mirrors the backward sweep of
the forward filtering backward smoothing recursion which, if it could
be calculated, would yield exactly $\bm{\psi}^{*}$.

\begin{algorithm}[H]
\caption{Recursive function approximations\label{alg:Function-approximations}}

\medskip{}
For $t=T,\ldots,1$:
\begin{enumerate}
\item \noindent Set $\psi_{t}^{i}\leftarrow g\left(\xi_{t}^{i},y_{t}\right)f\left(\xi_{t}^{i},\psi_{t+1}\right)$
for $i\in\{1,\ldots,N\}$.
\item \noindent Choose $\psi_{t}$ as a member of $\Psi$ on the basis of
$\xi_{t}^{1:N}$ and $\psi_{t}^{1:N}$.\end{enumerate}
\end{algorithm}

One choice in step 2. of Algorithm~\ref{alg:Function-approximations}
is to define $\psi_{t}$ using a non-parametric approximation such
as a Nadaraya--Watson estimate \citep{nadaraya1964estimating,watson1964smooth}.
Alternatively, a parametric approach is to choose $\psi_{t}$ as the
minimizer in some subset of $\Psi$ of some function of $\psi_{t}$,
$\xi_{t}^{1:N}$ and $\psi_{t}^{1:N}$. Although a number of choices
are possible, we focus in Section~\ref{sec:Applications-and} on
a simple parametric approach that is computationally inexpensive.

\subsection{The iterated auxiliary particle filter}

The iterated auxiliary particle filter (iAPF), Algorithm~\ref{alg:iAPF},
is obtained by iteratively running a $\bm{\psi}$-APF and estimating
$\bm{\psi}^{*}$ from its output. Specifically, after each $\bm{\psi}$-APF
is run, $\bm{\psi}^{*}$ is re-approximated using the particles obtained,
and the number of particles is increased according to a well-defined
rule. The algorithm terminates when a stopping rule is satisfied.

\begin{algorithm}[H]
\caption{An iterated auxiliary particle filter with parameters $(N_{0},k,\tau)$\label{alg:iAPF}}

\begin{enumerate}
\item Initialize: set $\bm{\psi}^{0}$ to be a sequence of constant functions,
$l\leftarrow0$.
\item \noindent Repeat:

\begin{enumerate}
\item \noindent Run a $\bm{\psi}^{l}$-APF with $N_{l}$ particles, and
set $\hat{Z}_{l}\leftarrow Z_{\bm{\psi}^{l}}^{N_{l}}$.
\item \noindent If $l>k$ and ${\rm sd}(\hat{Z}_{l-k:l})/{\rm mean}(\hat{Z}_{l-k:l})<\ensuremath{\tau}$,
go to 3.
\item \noindent Compute $\bm{\psi}^{l+1}$ using a version of Algorithm~\ref{alg:Function-approximations}
with the particles produced.
\item If $N_{l-k}=N_{l}$ and the sequence $\hat{Z}_{l-k:l}$ is not monotonically
increasing, set $N_{l+1}\leftarrow2N_{l}$. Otherwise, set $N_{l+1}\leftarrow N_{l}$.
\item Set $l\leftarrow l+1$ and go back to 2a.
\end{enumerate}
\item \noindent Run a $\bm{\psi}^{l}$-APF and return $\hat{Z}:=Z_{\bm{\psi}}^{N_{l}}$\end{enumerate}
\end{algorithm}

The rationale for step 2(d) of Algorithm~\ref{alg:iAPF} is that
if the sequence $\hat{Z}_{l-k:l}$ is monotonically increasing, there
is some evidence that the approximations $\bm{\psi}^{l-k:l}$ are
improving, and so increasing the number of particles may be unnecessary.
However, if the approximations $\hat{Z}_{l-k:l}$ have both high relative
standard deviation in comparison to $\tau$ and are oscillating then
reducing the variance of the approximation of $Z$ and/or improving
the approximation of $\bm{\psi}^{*}$ may require an increased number
of particles. Some support for this procedure can be obtained from
the log-normal CLT of \citet{EJP3428}: under regularity assumptions,
$\log Z_{\bm{\psi}}^{N}$ is approximately a $\mathcal{N}(-\delta_{\bm{\psi}}^{2}/2,\delta_{\bm{\psi}}^{2})$
random variable and so $\mathbb{P}\left(Z_{\bm{\psi}'}^{N}\geq Z_{\bm{\psi}}^{N}\right)\approx1-\Phi\left(\left[\delta_{\bm{\psi}'}^{2}-\delta_{\bm{\psi}}^{2}\right]/\left[2\sqrt{\delta_{\bm{\psi}}^{2}+\delta_{\bm{\psi}'}^{2}}\right]\right)$,
which is close to $1$ when $\delta_{\bm{\psi}'}^{2}\ll\delta_{\bm{\psi}}^{2}$.

\section{Approximations of smoothing expectations}

Thus far, we have focused on approximations of the marginal likelihood,
$L$, associated with a particular model and data record $y_{1:T}$.
Particle filters are also used to approximate so-called smoothing
expectations, i.e. $\pi(\varphi):=\mathbb{E}\left[\varphi(X_{1:T})\mid\{Y_{1:T}=y_{1:T}\}\right]$
for some $\varphi:\mathsf{X^{T}}\to\mathbb{R}$. Such approximations
can be motivated by a slight extension of (\ref{eq:Z}),
\[
\gamma(\varphi):=\int_{\mathsf{X}^{T}}\varphi(x_{1:T})\mu_{1}\left(x_{1}\right)g_{1}\left(x_{1}\right)\prod_{t=2}^{T}f_{t}\left(x_{t-1},x_{t}\right)g_{t}\left(x_{t}\right)dx_{1:T},
\]
where $\varphi$ is a real-valued, bounded, continuous function. We
can write $\pi(\varphi)=\gamma(\varphi)/\gamma(1)$, where $1$ denotes
the constant function $x\mapsto1$. We define below a well-known,
unbiased and strongly consistent estimate $\gamma^{N}(\varphi)$ of
$\gamma(\varphi)$, which can be obtained from Algorithm~\ref{alg:Particle-Filter}.
A strongly consistent approximation of $\pi(\varphi)$ can then be
defined as $\gamma^{N}(\varphi)/\gamma^{N}(1)$.

The definition of $\gamma^{N}(\varphi)$ is facilitated by a specific
implementation of step 2. of Algorithm~\ref{alg:Particle-Filter}
in which one samples
\[
A_{t-1}^{i}\sim{\rm Categorical}\left(\frac{g_{t-1}(\xi_{t-1}^{1})}{\sum_{j=1}^{N}g_{t-1}(\xi_{t-1}^{j})},\ldots,\frac{g_{t-1}(\xi_{t-1}^{N})}{\sum_{j=1}^{N}g_{t-1}(\xi_{t-1}^{j})}\right),\qquad\xi_{t}^{i}\sim f_{t}(\xi_{t-1}^{A_{t-1}^{i}},\cdot),
\]
for each $i\in\{1,\ldots,N\}$ independently. Use of, e.g., the Alias
algorithm \citep{walker1974new,walker1977efficient} gives the algorithm
$\mathcal{O}(N)$ computational complexity, and the random variables
$\left(A_{t}^{i};t\in\{1,\ldots,T-1\},i\in\{1,\ldots,N\}\right)$
provide ancestral information associated with each particle. By defining
recursively for each $i\in\{1,\ldots,N\}$, $B_{T}^{i}:=i$ and $B_{t-1}^{i}:=A_{t-1}^{B_{t}^{i}}$
for $t=T-1,\ldots,1$, the $\{1,\ldots,N\}^{T}$-valued random variable
$B_{1:T}^{i}$ encodes the ancestral lineage of $\xi_{T}^{i}$ \citep{Andrieu2010}.
It follows from \citet[Theorem~7.4.2]{DelMoral2004} that the approximation
\[
\gamma^{N}(\varphi):=\left[\frac{1}{N}\sum_{i=1}^{N}g_{T}(\xi_{T}^{i})\varphi(\xi_{1}^{B_{1}^{i}},\xi_{2}^{B_{2}^{i}},\dots,\xi_{T}^{B_{T}^{i}})\right]\prod_{t=1}^{T-1}\left(\frac{1}{N}\sum_{i=1}^{N}g_{t}(\xi_{t}^{i})\right),
\]
is unbiased and strongly consistent, and a strongly consistent approximation
of $\pi(\varphi)$ is
\begin{equation}
\pi^{N}(\varphi):=\frac{\gamma^{N}(\varphi)}{\gamma^{N}(1)}=\frac{1}{\sum_{i=1}^{N}g_{T}(\xi_{T}^{i})}\sum_{i=1}^{N}\varphi\left(\xi_{1}^{B_{1}^{i}},\xi_{2}^{B_{2}^{i}},\dots,\xi_{T}^{B_{T}^{i}}\right)g_{T}(\xi_{T}^{i}).\label{eq:smoothest-1}
\end{equation}
The $\bm{\psi}^{*}$-APF is optimal in terms of approximating $\gamma(1)\equiv Z$
and not $\pi(\varphi)$ for general $\varphi$. Asymptotic variance
expressions akin to Proposition~\ref{prop:clt}, but for $\pi_{\bm{\psi}}^{N}(\varphi)$,
can be derived using existing results \citep[see, e.g.,][]{del1999central,Chopin2004,kunsch2005recursive,Douc08weighted}
in the same manner. These could be used to investigate the influence
of $\bm{\psi}$ on the accuracy of $\pi_{\bm{\psi}}^{N}(\varphi)$
or the interaction between $\varphi$ and the sequence $\bm{\psi}$
which minimizes the asymptotic variance of the estimator of its expectation.

Finally, we observe that when the optimal sequence $\bm{\psi}^{*}$
is used in an APF in conjunction with an adaptive resampling strategy
(see Algorithm~\ref{alg:apf_psi-ar} below), the weights are all
equal, no resampling occurs and the $\xi_{t}^{i}$ are all i.i.d.
samples from $\mathbb{P}\left(X_{t}\in\cdot\mid\{Y_{1:T}=y_{1:T}\}\right)$.
This at least partially justifies the use of iterated $\bm{\psi}$-APFs
to approximate $\bm{\psi}^{*}$: the asymptotic variance $\sigma_{\bm{\psi}}^{2}$
in (\ref{eq:avarpsi}) is particularly affected by discrepancies between
$\bm{\psi}^{*}$ and $\bm{\psi}$ in regions of relatively high conditional
probability given the data record $y_{1:T}$, which is why we have
chosen to use the particles as support points to define approximations
of $\bm{\psi}^{*}$ in Algorithm~\ref{alg:Function-approximations}.

\section{Applications and examples\label{sec:Applications-and}}

The purpose of this section is to demonstrate that the iAPF can provide
substantially better estimates of the marginal likelihood $L$ than
the BPF at the same computational cost. This is exemplified by its
performance when $d$ is large, recalling that $\mathsf{X}=\mathbb{R}^{d}$.
When $d$ is large, the BPF typically requires a large number of particles
in order to approximate $L$ accurately. In contrast, the $\bm{\psi}^{*}$-APF
computes $L$ exactly, and we investigate below the extent to which
the iAPF is able to provide accurate approximations in this setting.
Similarly, when there are unknown statistical parameters $\theta$,
we show empirically that the accuracy of iAPF approximations of the
likelihood $L(\theta)$ are more robust to changes in $\theta$ than
their BPF counterparts.

Unbiased, non-negative approximations of likelihoods $L(\theta)$
are central to the particle marginal Metropolis--Hastings algorithm
(PMMH) of \citet{Andrieu2010}, a prominent parameter estimation algorithm
for general state space hidden Markov models. An instance of a pseudo-marginal
Markov chain Monte Carlo algorithm \citep{Beaumont2003,Andrieu2009},
the computational efficiency of PMMH depends, sometimes dramatically,
on the quality of the unbiased approximations of $L(\theta)$ \citep{andrieu2015,Leea,sherlock2015,doucet2015efficient}
delivered by a particle filter for a range of $\theta$ values. The
relative robustness of iAPF approximations of $L(\theta)$ to changes
in $\theta$, mentioned above, motivates their use over BPF approximations
in PMMH.

\subsection{Implementation details\label{sub:Implementation-details}}

In our examples, we use a parametric optimization approach in Algorithm~3.
Specifically, for each $t\in\{1,\ldots,T\}$, we compute numerically
\begin{equation}
\left(m_{t}^{*},\Sigma_{t}^{*},\lambda_{t}^{*}\right)=\mbox{argmin}_{\left(m,\Sigma,\lambda\right)}\sum_{i=1}^{N}\left[\mathcal{N}\left(\xi_{t}^{i};m,\Sigma\right)-\lambda\psi_{t}^{i}\right]^{2},\label{eq:param}
\end{equation}
and then set
\begin{equation}
\psi_{t}(x_{t}):=\mathcal{N}\left(x_{t};m_{t}^{*},\Sigma_{t}^{*}\right)+c(N,m_{t}^{*},\Sigma_{t}^{*}),\label{eq:psi_param}
\end{equation}
where $c$ is a positive real-valued function, which ensures that
$f_{t}^{\bm{\psi}}(x,\cdot)$ is a mixture of densities with some
non-zero weight associated with the mixture component $f(x,\cdot)$.
This is intended to guard against terms in the asymptotic variance
$\sigma_{\bm{\psi}}^{2}$ in (\ref{eq:avarpsi}) being very large
or unbounded. We chose (\ref{eq:param}) for simplicity and its low
computational cost, and it provided good performance in our simulations.
For the stopping rule, we used $k=5$ for the application in Section~\ref{sub:Linear-Gaussian-model},
and $k=3$ for the applications in Sections~\ref{sub:Univariate-Stochastic-Volatility}--\ref{sub:Multivariate-Stochastic-Volatili}.
We observed empirically that the relative standard deviation of the
likelihood estimate tended to be close to, and often smaller than,
the chosen level for $\tau$. A value of $\tau=1$ should therefore
be sufficient to keep the relative standard deviation around 1 as
desired \citep[see, e.g.,][]{doucet2015efficient,sherlock2015}. We
set $\tau=0.5$ as a conservative choice for all our simulations apart
from the multivariate stochastic volatility model of Section~\ref{sub:Multivariate-Stochastic-Volatili},
where we set $\tau=1$ to improve speed. We performed the minimization
in (\ref{eq:param}) under the restriction that $\Sigma$ was a diagonal
matrix, as this was considerably faster and preliminary simulations
suggested that this was adequate for the examples considered.

We used an effective sample size based resampling scheme \citep{Kong1994,Liu1995},
described in Algorithm~\ref{alg:apf_psi-ar} with a user-specified
parameter $\kappa\in[0,1]$. The effective sample size is defined
as ${\rm ESS}(W^{1},\ldots,W^{N}):=\left(\sum_{i=1}^{N}W^{i}\right)^{2}/\sum_{i=1}^{N}\left(W^{i}\right)^{2}$,
and the estimate of $Z$ is
\[
Z^{N}:=\prod_{t\in\mathcal{R}\cup\{T\}}\left[\frac{1}{N}\sum_{i=1}^{N}W_{t}^{i}\right],\qquad\mathcal{R}:=\left\{ t\in\{1,\ldots,T-1\}:{\rm ESS}(W_{t}^{1},\ldots,W_{t}^{N})\leq\kappa N\right\} .
\]
where $\mathcal{R}$ is the set of ``resampling times''. This reduces
to Algorithm~\ref{alg:apf_psi} when $\kappa=1$ and to a simple
importance sampling algorithm when $\kappa=0$; we use $\kappa=0.5$
in our simulations. The use of adaptive resampling is motivated by
the fact that when the effective sample size is large, resampling
can be detrimental in terms of the quality of the approximation $Z^{N}$.

\begin{algorithm}
\caption{$\bm{\psi}$-Auxiliary Particle Filter with $\kappa$-adaptive resampling\label{alg:apf_psi-ar}}

\begin{enumerate}
\item Sample $\xi_{1}^{i}\sim\mu_{1}^{\bm{\psi}}$ independently, and set
$W_{1}^{i}\leftarrow g_{1}^{\bm{\psi}}(\xi_{1}^{i})$ for $i\in\{1,\ldots,N\}$.
\item For $t=2,\ldots,T$:

\begin{enumerate}
\item If ${\rm ESS}(W_{t-1}^{1},\ldots,W_{t-1}^{N})\leq\kappa N$, sample
independently
\[
\xi_{t}^{i}\sim\frac{\sum_{j=1}^{N}W_{t-1}^{j}f_{t}^{\bm{\psi}}(\xi_{t-1}^{j},\cdot)}{\sum_{j=1}^{N}W_{t-1}^{j}},\qquad i\in\{1,\ldots,N\},
\]
and set $W_{t}^{i}\leftarrow g_{t}^{\bm{\psi}}(\xi_{t}^{i})$, $i\in\{1,\ldots,N\}$.
\item Otherwise, sample $\xi_{t}^{i}\sim f_{t}^{\bm{\psi}}(\xi_{t-1}^{i},\cdot)$
independently, and set $W_{t}^{i}\leftarrow W_{t-1}^{i}g_{t}^{\bm{\psi}}(\xi_{t}^{i})$
for $i\in\{1,\ldots,N\}$.\end{enumerate}
\end{enumerate}
\end{algorithm}

\subsection{Linear Gaussian model\label{sub:Linear-Gaussian-model}}

A linear Gaussian HMM is defined by the following initial, transition
and observation Gaussian densities: $\mu(\cdot)=\mathcal{N}\left(\cdot;m,\Sigma\right)$,
$f(x,\cdot)=\mathcal{N}\left(\cdot;Ax,B\right)$ and $g(x,\cdot)=\mathcal{N}\left(\cdot;Cx,D\right)$,
where $m\in\mathbb{R}^{d}$, $\Sigma,A,B\in\mathbb{R}^{d\times d}$,
$C\in\mathbb{R}^{d\times d'}$ and $D\in\mathbb{R}^{d'\times d'}$.
For this model, it is possible to implement the fully adapted APF
(FA-APF) and to compute explicitly the marginal likelihood, filtering
and smoothing distributions using the Kalman filter, facilitating
comparisons. We emphasize that implementation of the FA-APF is possible
only for a restricted class of analytically tractable models, while
the iAPF methodology is applicable more generally. Nevertheless, the
iAPF exhibited better performance than the FA-APF in our examples.

\subsubsection*{Relative variance of approximations of $Z$ when $d$ is large}

We consider a family of Linear Gaussian models where $m={\bf 0}$,
$\Sigma=B=C=D=I_{d}$ and $A_{ij}=\alpha^{|i-j|+1}$, $i,j\in\{1,\ldots,d\}$
for some $\alpha\in\left(0,1\right)$. Our first comparison is between
the relative errors of the approximations $\hat{Z}$ of $L=Z$ using
the iAPF, the BPF and the FA-APF. We consider configurations with
$d\in\left\{ 5,10,20,40,80\right\} $ and $\alpha=0.42$ and we simulated
a sequence of $T=100$ observations $y_{1:T}$ for each configuration.
We ran $1000$ replicates of the three algorithms for each configuration
and report box plots of the ratio $\hat{Z}/Z$ in Figure~\ref{fig:Box-plots-of}.

\begin{figure}[H]
\noindent \centering{}\includegraphics{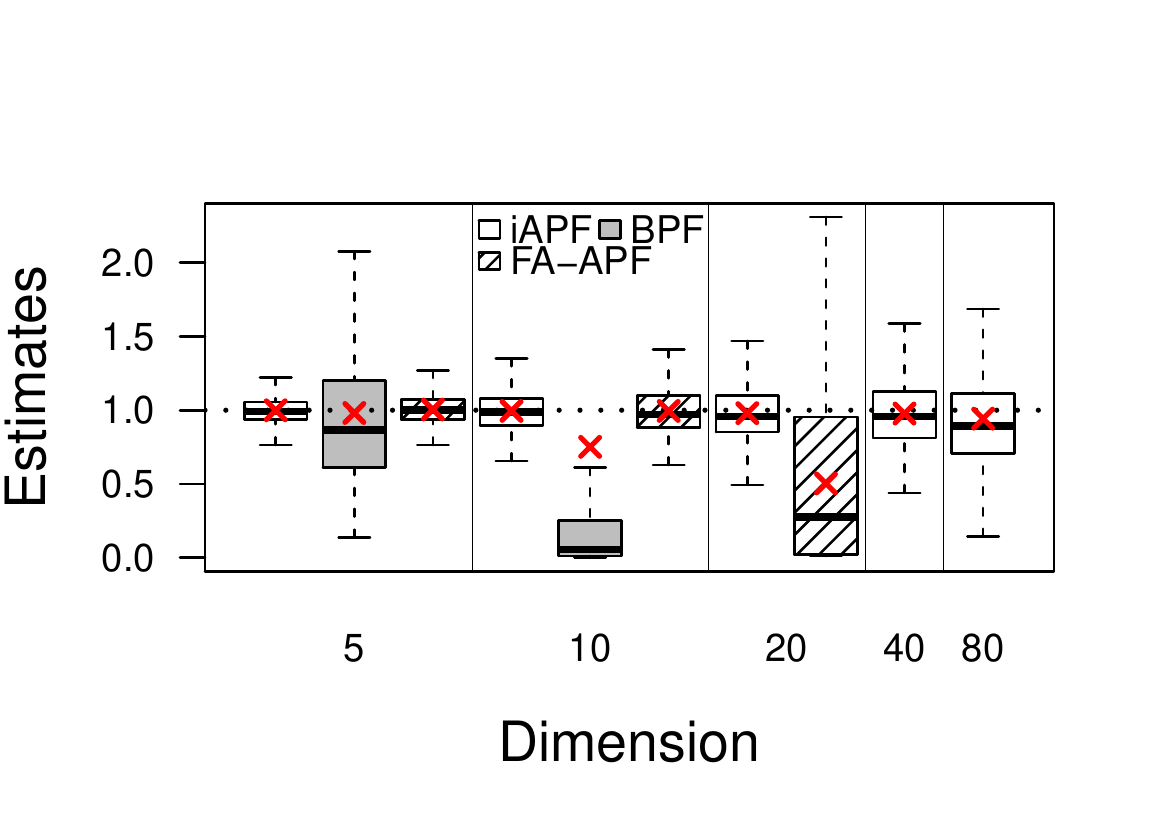}\caption{\label{fig:Box-plots-of}Box plots of $\hat{Z}/Z$ for different dimensions
using $1000$ replicates. The crosses indicate the mean of each sample.}
\end{figure}

For all the simulations we ran an iAPF with $N_{0}=1000$ starting
particles, a BPF with $N=10000$ particles and an FA-APF with $N=5000$
particles. The BPF and FA-APF both had slightly larger average computational
times than the iAPF with these configurations. The average number
of particles for the final iteration of the iAPF was greater than
$N_{0}$ only in dimensions $d=40$ ($1033$) and $d=80$ ($1142$).
For $d>10$, it was not possible to obtain reasonable estimates with
the BPF in a feasible computational time (similarly for the FA-APF
for $d>20$). The standard deviation of the samples and the average
resampling count across the chosen set of dimensions are reported
in Tables~\ref{tab:Empirical-standard-deviation}--\ref{tab:Average-resampling-count}.

\begin{table}[H]
\caption{\label{tab:Empirical-standard-deviation}Empirical standard deviation
of the quantity $\hat{Z}/Z$ using $1000$ replicates}

\noindent \centering{}%
\begin{tabular}{|c|c|c|c|c|c|}
\hline 
Dimension & $5$ & $10$ & $20$ & $40$ & $80$\tabularnewline
\hline 
\hline 
iAPF & $0.09$ & $0.14$ & $0.19$ & $0.23$ & $0.35$\tabularnewline
\hline 
BPF & $0.51$ & $6.4$ & - & - & -\tabularnewline
\hline 
FA-APF & $0.10$ & $0.17$ & $0.53$ & - & -\tabularnewline
\hline 
\end{tabular}
\end{table}
\begin{table}[H]
\caption{\label{tab:Average-resampling-count}Average resampling count for
the $1000$ replicates}

\noindent \centering{}%
\begin{tabular}{|c|c|c|c|c|c|}
\hline 
Dimension & $5$ & $10$ & $20$ & $40$ & $80$\tabularnewline
\hline 
\hline 
iAPF & $6.93$ & $15.11$ & $27.61$ & $42.41$ & $71.88$\tabularnewline
\hline 
BPF & $99$ & $99$ & - & - & -\tabularnewline
\hline 
FA-APF & $26.04$ & $52.71$ & $84.98$ & - & -\tabularnewline
\hline 
\end{tabular}
\end{table}

Fixing the dimension $d=10$ and the simulated sequence of observations
$y_{1:T}$ with $\alpha=0.42$, we now consider the variability of
the relative error of the estimates of the marginal likelihood of
the observations using the iAPF and the BPF for different values of
the parameter $\alpha\in\{0.3,0.32,\dots,0.48,0.5\}$. In Figure~\ref{fig:Box-plots-of-1},
we report box plots of $\hat{Z}/Z$ in $1000$ replications. For the
iAPF, the length of the boxes are significantly less variable across
the range of values of $\alpha$. In this case, we used $N=50000$
particles for the BPF, giving a computational time at least five times
larger than that of the iAPF. This demonstrates that the approximations
of the marginal likelihood $L(\alpha)$ provided by the iAPF are relatively
insensitive to small changes in $\alpha$, in contrast to the BPF.
Similar simulations, which we do not report, show that the FA-APF
for this problem performs slightly worse than the iAPF at double the
computational time.

\begin{figure}[H]
\noindent \centering{}\subfloat[iAPF]{

\centering{}\includegraphics[width=8cm,height=7cm]{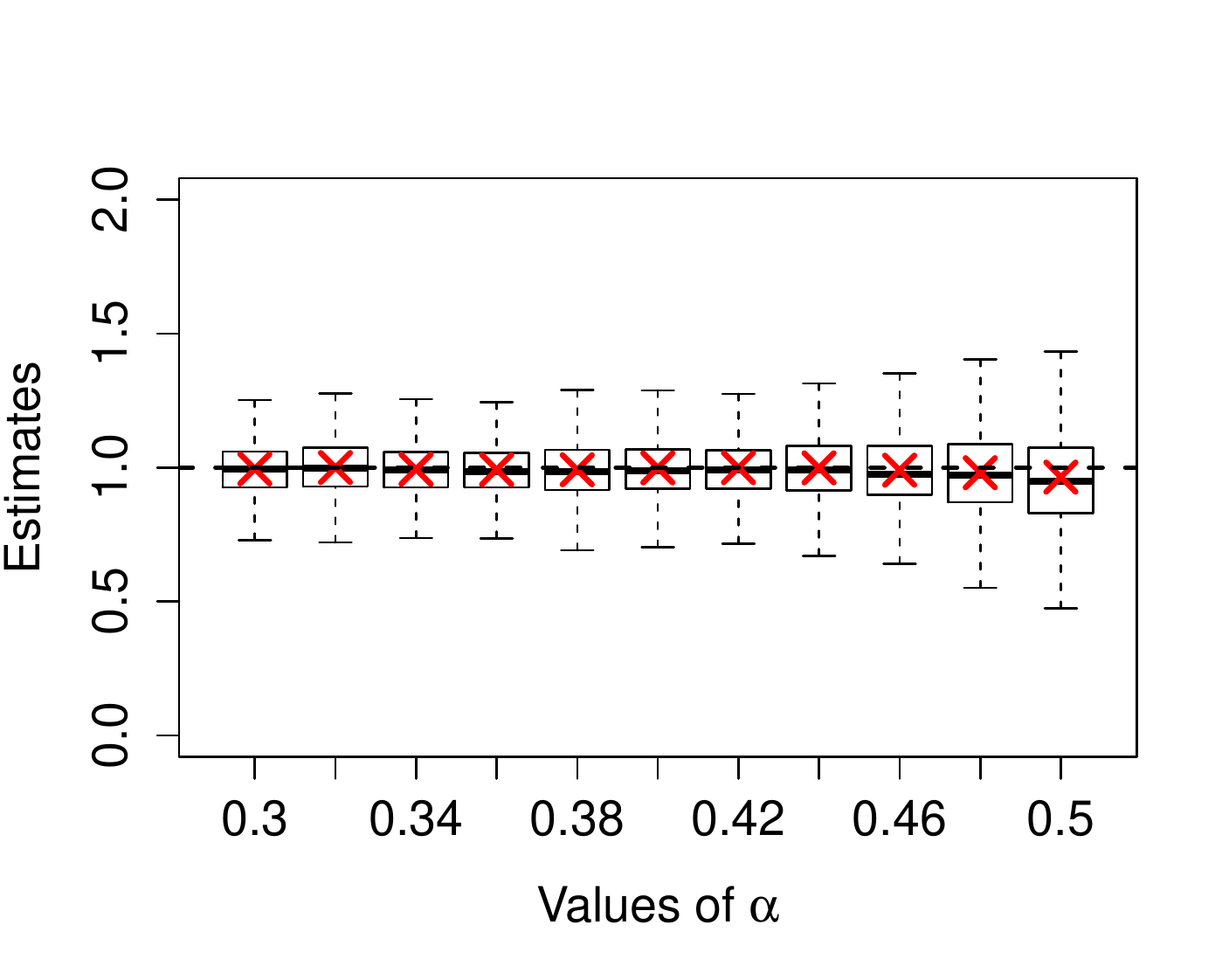}}\subfloat[BPF]{\noindent \centering{}\includegraphics[width=8cm,height=7cm]{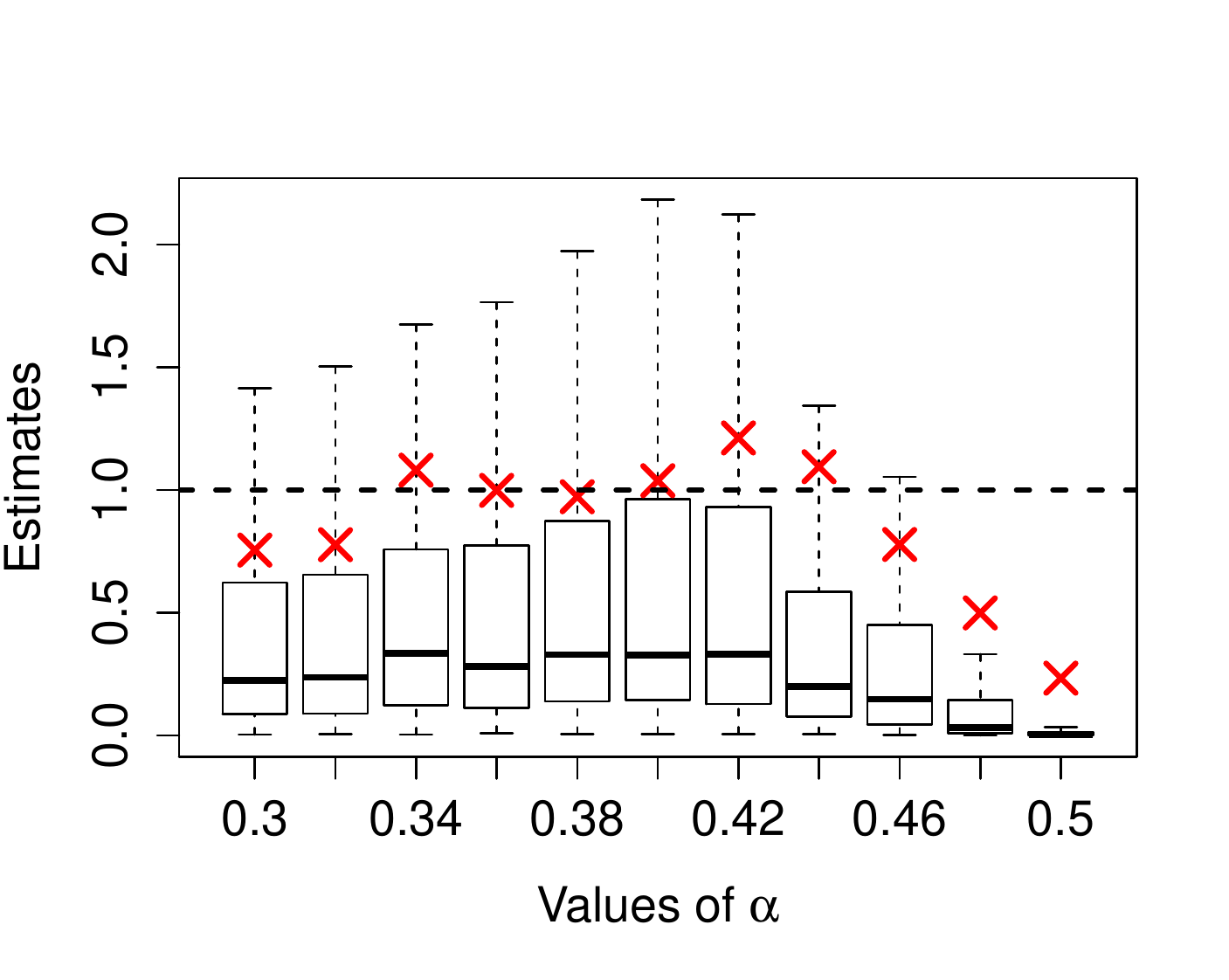}}\caption{\label{fig:Box-plots-of-1}Box plots of $\hat{Z}/Z$ for different
values of the parameter $\alpha$ using $1000$ replicates. The crosses
indicate the mean of each sample.}
\end{figure}

\subsubsection*{Particle marginal Metropolis--Hastings}

We consider a Linear Gaussian model with $m={\bf 0}$, $\Sigma=B=C=I_{d}$,
and $D=\delta I_{d}$ with $\delta=0.25$. We used the lower-triangular
matrix
\[
A=\left(\begin{array}{ccccc}
0.9 & 0 & 0 & 0 & 0\\
0.3 & 0.7 & 0 & 0 & 0\\
0.1 & 0.2 & 0.6 & 0 & 0\\
0.4 & 0.1 & 0.1 & 0.3 & 0\\
0.1 & 0.2 & 0.5 & 0.2 & 0
\end{array}\right),
\]
and simulated a sequence of $T=100$ observations. Assuming only that
$A$ is lower triangular, for identifiability, we performed Bayesian
inference for the 15 unknown parameters $\left\{ A_{i,j}:i,j\in\left\{ 1,\dots,5\right\} ,j\leq i\right\} $,
assigning each parameter an independent uniform prior on $[-5,5]$.
From the initial point $A_{1}=I_{5}$ we ran three Markov chains $A_{1:L}^{{\rm BPF}}$,
$A_{1:L}^{{\rm iAPF}}$ and $A_{1:L}^{{\rm Kalman}}$ of length $L=300000$
to explore the parameter space, updating one of the $15$ parameters
components at a time with a Gaussian random walk proposal with variance
$0.1$. The chains differ in how the acceptance probabilities are
computed, and correspond to using unbiased estimates of the marginal
likelihood obtain from the BPF, iAPF or the Kalman filter, respectively.
In the latter case, this corresponds to running a Metropolis--Hastings
(MH) chain by computing the marginal likelihood exactly. We started
every run of the iAPF with $N_{0}=500$ particles. The resulting average
number of particles used to compute the final estimate was $500.2$.
The number of particles $N=20000$ for the BPF was set to have a greater
computational time, in this case $A_{1:L}^{{\rm BPF}}$ took $50\%$
more time than $A_{1:L}^{{\rm iAPF}}$ to simulate.

In Figure~\ref{fig:Linear-Gaussian-model:}, we plot posterior density
estimates obtained from the three chains for 3 of the 15 entries of
the transition matrix $A$. The posterior means associated with the
entries of the matrix $A$ were fairly close to $A$ itself, the largest
discrepancy being around $0.2$, and the posterior standard deviations
were all around $0.1$. A comparison of estimated Markov chain autocorrelations
for these same parameters is reported in Figure~\ref{fig:Linear-Gaussian-model:-1},
which indicates little difference between the iAPF-PMMH and Kalman-MH
Markov chains, and substantially worse performance for the BPF-PMMH
Markov chain. The integrated autocorrelation time of the Markov chains
provides a measure of the asymptotic variance of the individual chains'
ergodic averages, and in this regard the iAPF-PMMH and Kalman-MH Markov
chains were practically indistinguishable, while the BPF-PMMH performed
between 3 and 4 times worse, depending on the parameter. The relative
improvement of the iAPF over the BPF does seem empirically to depend
on the value of $\delta$. In experiments with larger $\delta$, the
improvement was still present but less pronounced than for $\delta=0.25$.
We note that in this example, $\bm{\psi}^{*}$ is outside the class
of possible $\bm{\psi}$ sequences that can be obtained using the
iAPF: the approximations in $\bm{\Psi}$ are functions that are constants
plus a multivariate normal density with a diagonal covariance matrix
whilst the functions in $\bm{\psi}^{*}$ are multivariate normal densities
whose covariance matrices have non-zero, off-diagonal entries.

\begin{figure}[H]
\noindent \centering{}\subfloat[$A_{11}$]{\noindent \centering{}\includegraphics[width=5cm,height=5cm]{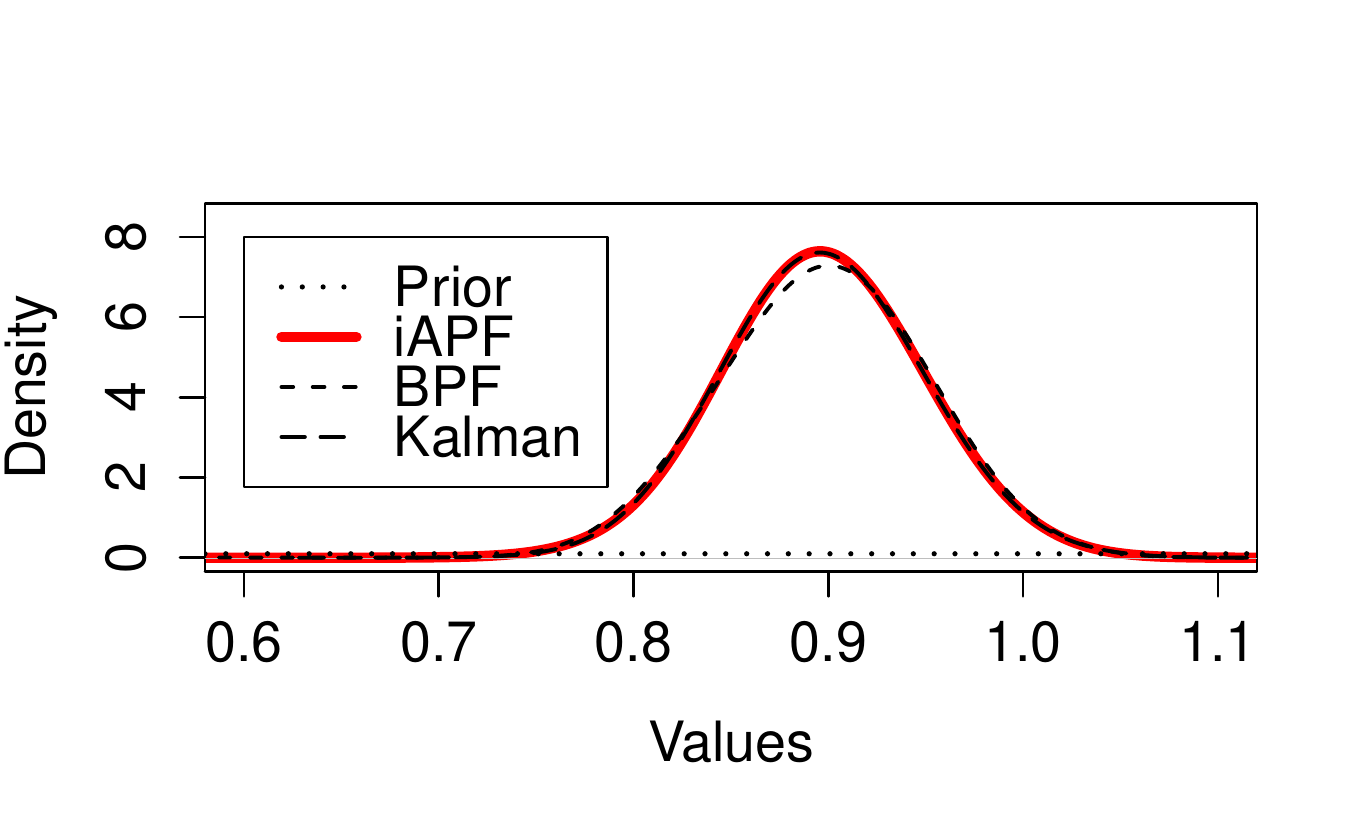}}\subfloat[$A_{41}$]{

\centering{}\includegraphics[width=5cm,height=5cm]{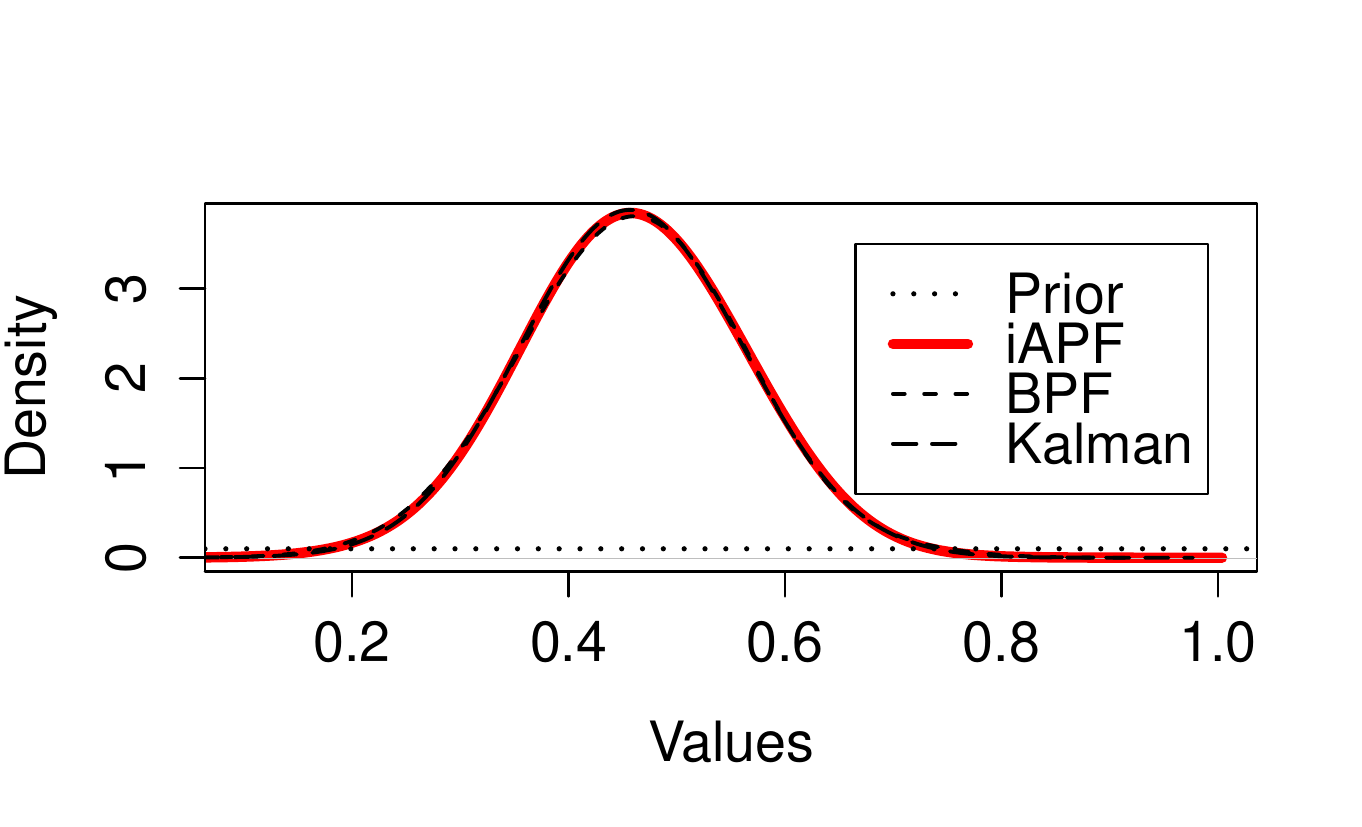}}\subfloat[$A_{55}$]{

\centering{}\includegraphics[width=5cm,height=5cm]{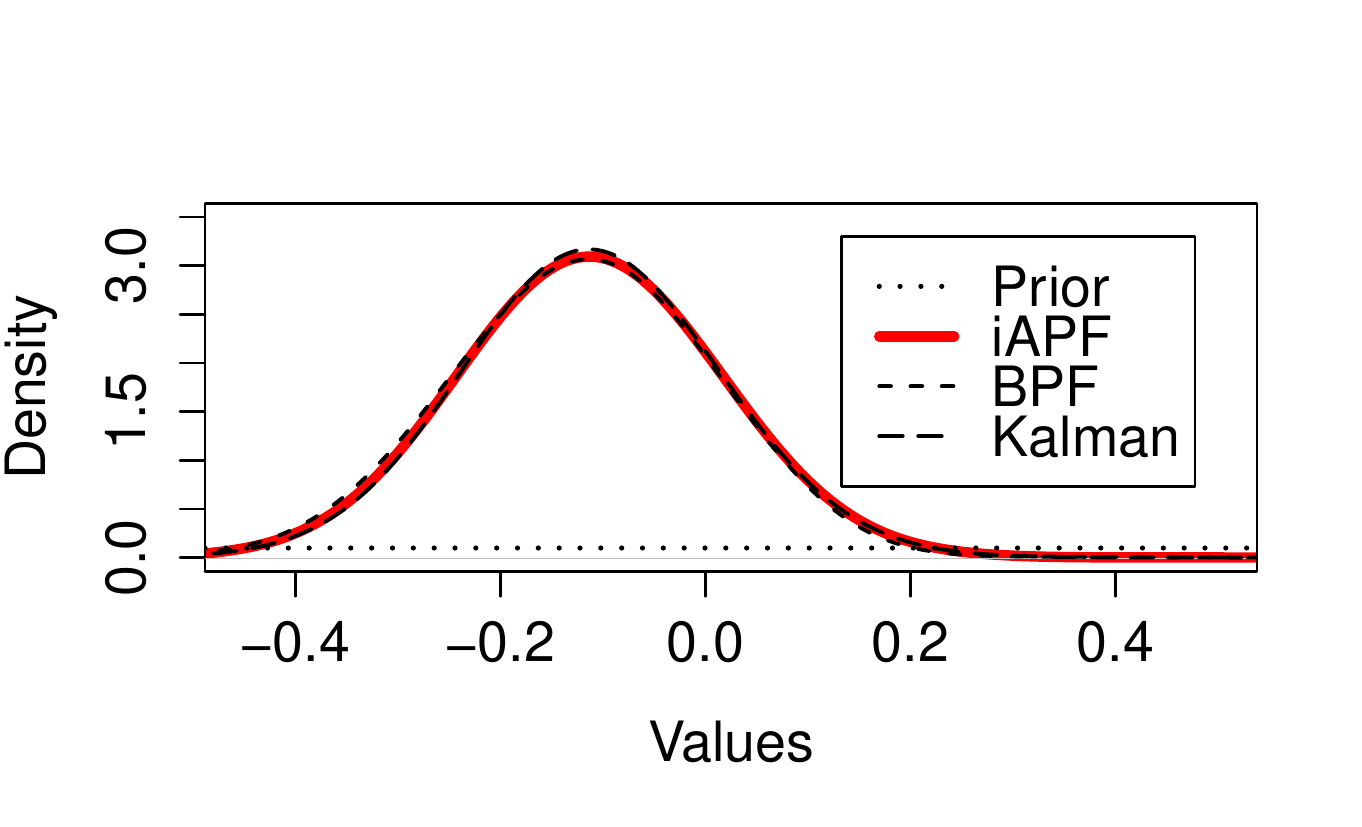}}\caption{\label{fig:Linear-Gaussian-model:}Linear Gaussian model: density
estimates for the specified parameters from the three Markov chains.}
\end{figure}

\begin{figure}[H]
\noindent \begin{centering}
\subfloat[$A_{11}$]{\noindent \centering{}\includegraphics[width=5cm,height=5cm]{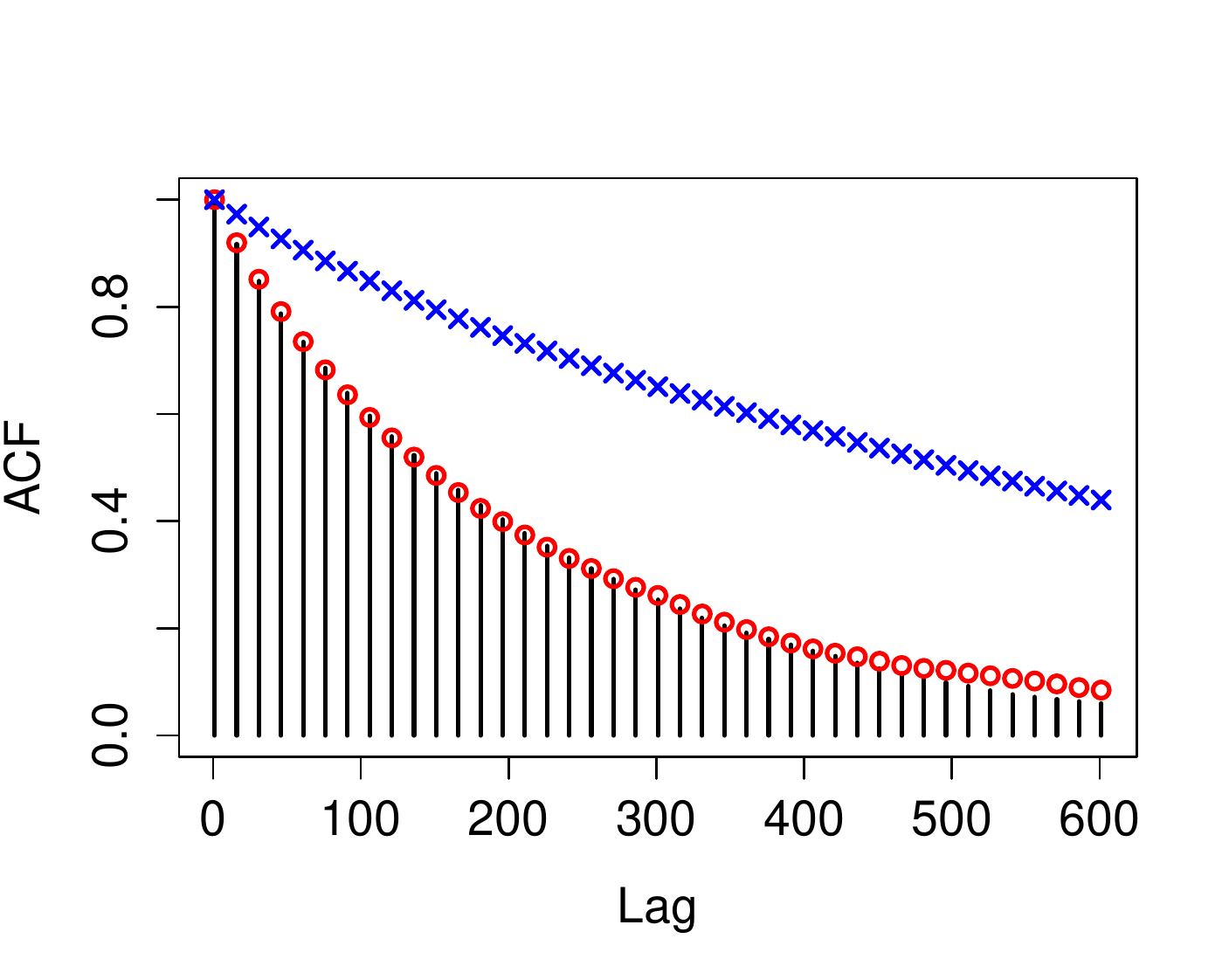}}\subfloat[$A_{41}$]{\centering{}\includegraphics[width=5cm,height=5cm]{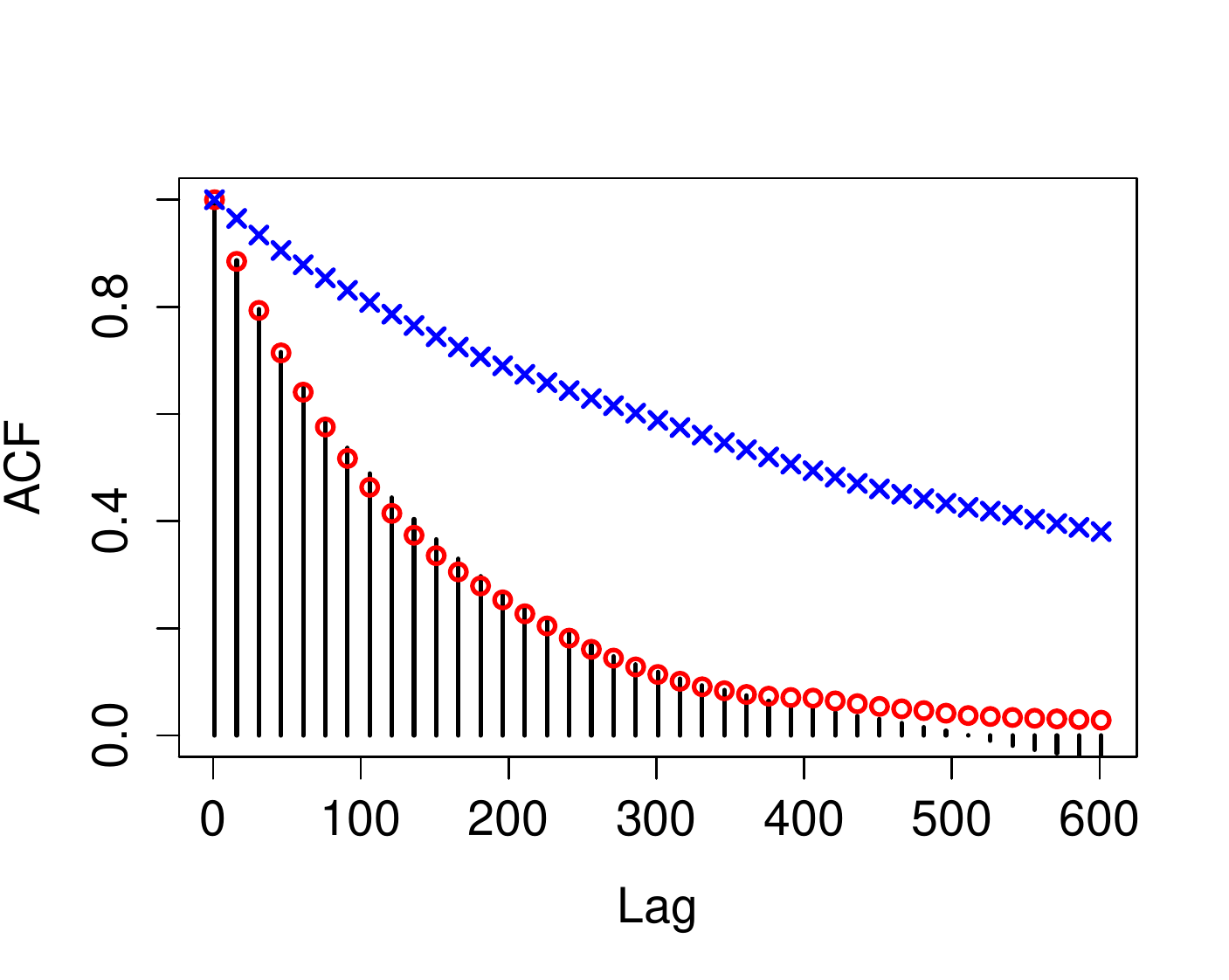}}\subfloat[$A_{55}$]{\begin{centering}
\includegraphics[width=5cm,height=5cm]{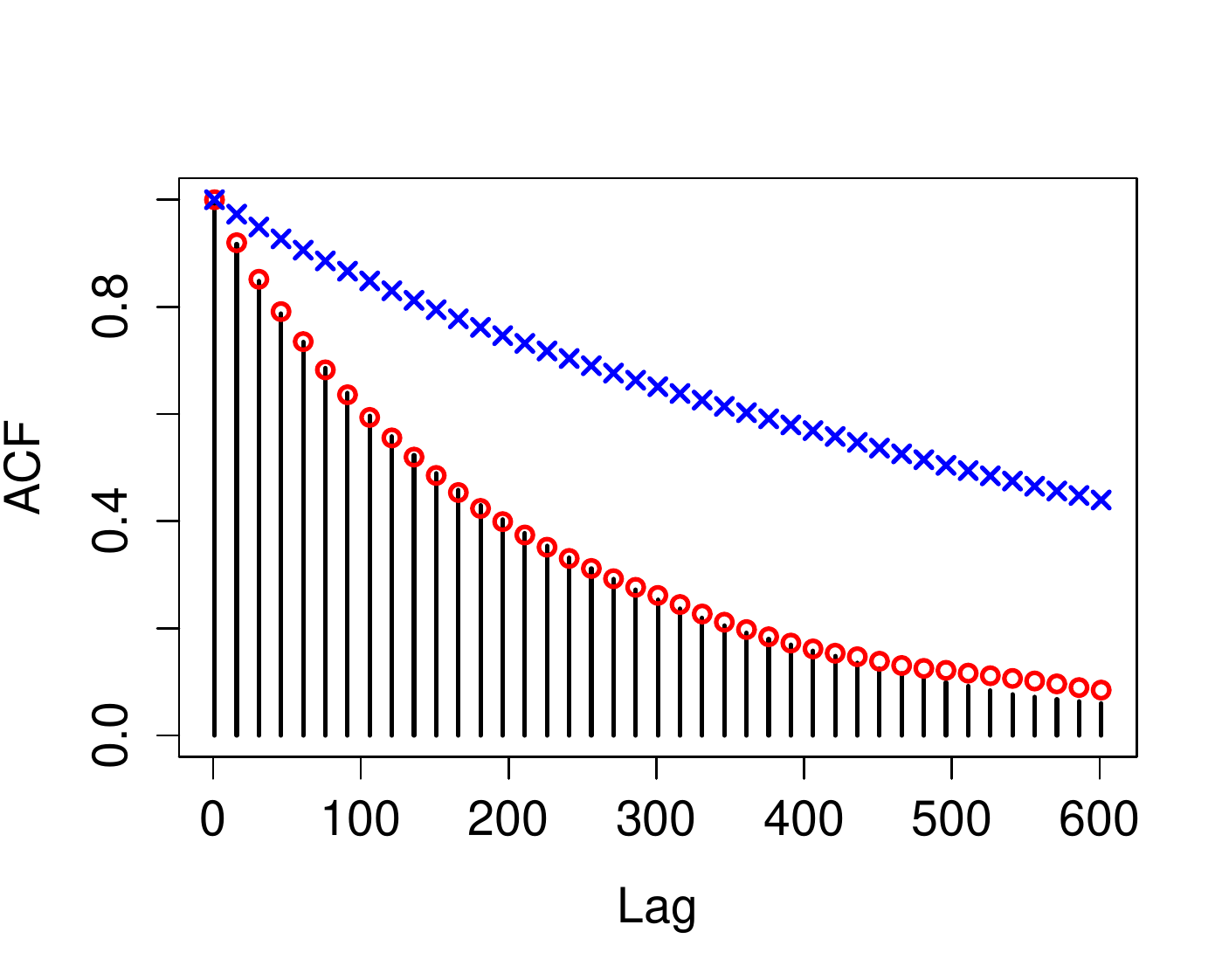}
\par\end{centering}

}
\par\end{centering}

\noindent \centering{}\caption{\label{fig:Linear-Gaussian-model:-1}Linear Gaussian model: autocorrelation
function estimates for the BPF-PMMH (crosses), iAPF-PMMH (solid lines)
and Kalman-MH (circles) Markov chains. }
\end{figure}

\subsection{Univariate stochastic volatility model\label{sub:Univariate-Stochastic-Volatility}}

A simple stochastic volatility model is defined by $\mu(\cdot)=\mathcal{N}(\cdot;0,\sigma^{2}/(1-\alpha)^{2})$,
$f(x,\cdot)=\mathcal{N}(\cdot;\alpha x,\sigma^{2})$ and $g(x,\cdot)=\mathcal{N}(\cdot;0,\beta^{2}\exp(x))$,
where $\alpha\in(0,1)$, $\beta>0$ and $\sigma^{2}>0$ are statistical
parameters \citep[see, e.g.,][]{Kim1998}. To compare the efficiency
of the iAPF and the BPF within a PMMH algorithm, we analyzed a sequence
of $T=945$ observations $y_{1:T}$, which are mean-corrected daily
returns computed from weekday close exchange rates $r_{1:T+1}$ for
the pound/dollar from 1/10/81 to 28/6/85. This data has been previously
analyzed using different approaches, e.g. in \citet{Harvey1994} and
\citet{Kim1998}. 

We wish to infer the model parameters $\theta=\left(\alpha,\sigma,\beta\right)$
using a PMMH algorithm and compare the two cases where the marginal
likelihood estimates are obtained using the iAPF and the BPF. We placed
independent inverse Gamma prior distributions $\mathcal{IG}\left(2.5,0.025\right)$
and $\mathcal{IG}\left(3,1\right)$ on $\sigma^{2}$ and $\beta^{2}$,
respectively, and an independent ${\rm Beta}\left(20,1.5\right)$
prior distribution on the transition coefficient $\alpha$. We used
$\left(\alpha_{0},\sigma_{0},\beta_{0}\right)=\left(0.95,\sqrt{0.02},0.5\right)$
as the starting point of the three chains: $X_{1:L}^{{\rm iAPF}}$,
$X_{1:L}^{{\rm BPF}}$ and $X_{L'}^{{\rm BPF}'}$. All the chains
updated one component at a time with a Gaussian random walk proposal
with variances $\left(0.02,0.05,0.1\right)$ for the parameters $\left(\alpha,\sigma,\beta\right)$.
$X_{1:L}^{{\rm iAPF}}$ has a total length of $L=150000$ and for
the estimates of the marginal likelihood that appear in the acceptance
probability we use the iAPF with $N_{0}=100$ starting particles.
For $X_{1:L}^{{\rm BPF}}$ and $X_{1:L'}^{{\rm BPF}'}$ we use BPFs:
$X_{1:L}^{{\rm BPF}'}$ is a shorter chain with more particles ($L=150000$
and $N=1000$) while $X_{1:L'}^{{\rm BPF}'}$ is a longer chain with
fewer particles $(L=1500000$, $N=100$). All chains required similar
running time overall to simulate. Figure~\ref{fig:Stochastic-Volatility-model:}
shows estimated marginal posterior densities for the three parameters
using the different chains.

\begin{figure}[H]
\noindent \centering{}\subfloat[$\alpha$]{\noindent \centering{}\includegraphics[width=5cm,height=5cm]{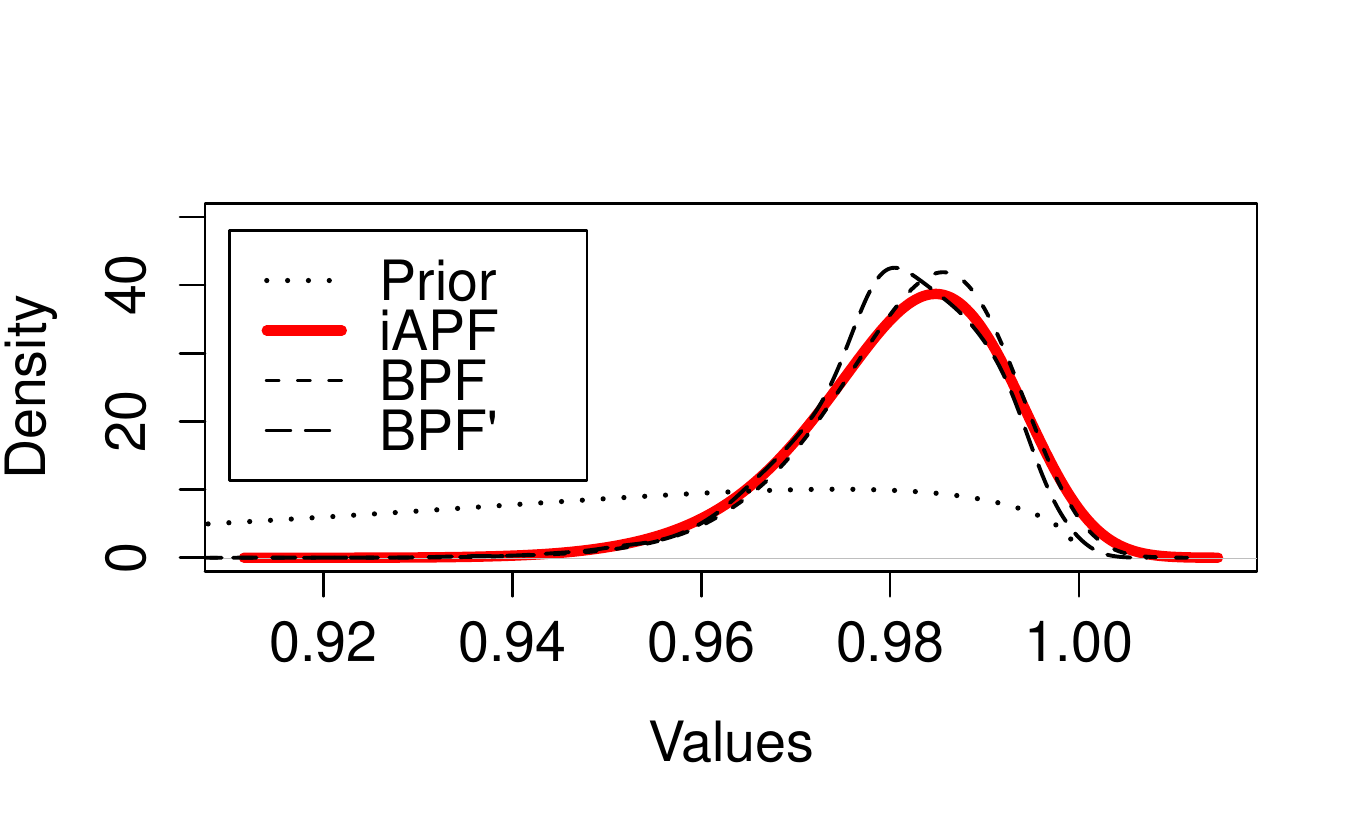}}\subfloat[$\sigma$]{\centering{}\includegraphics[width=5cm,height=5cm]{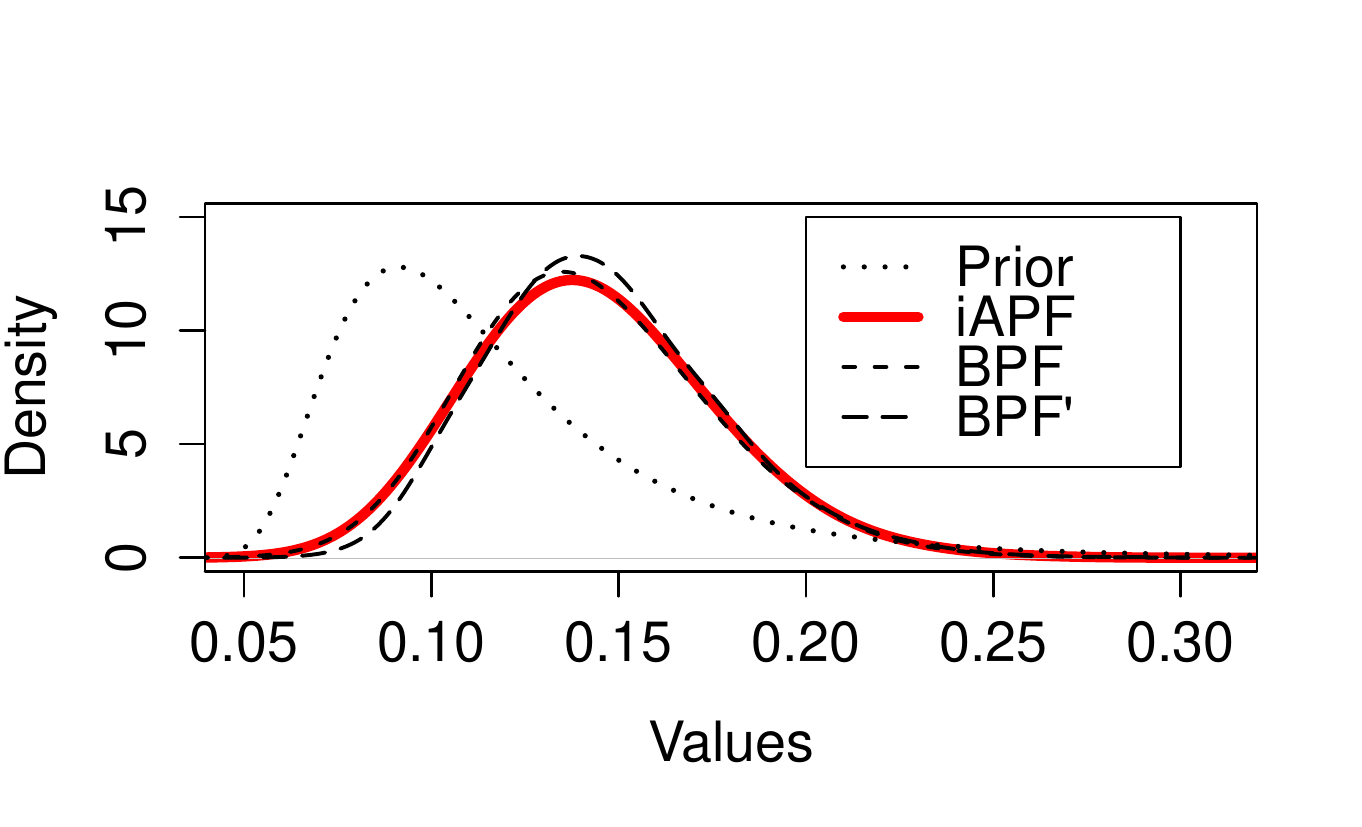}}\subfloat[$\beta$]{\centering{}\includegraphics[width=5cm,height=5cm]{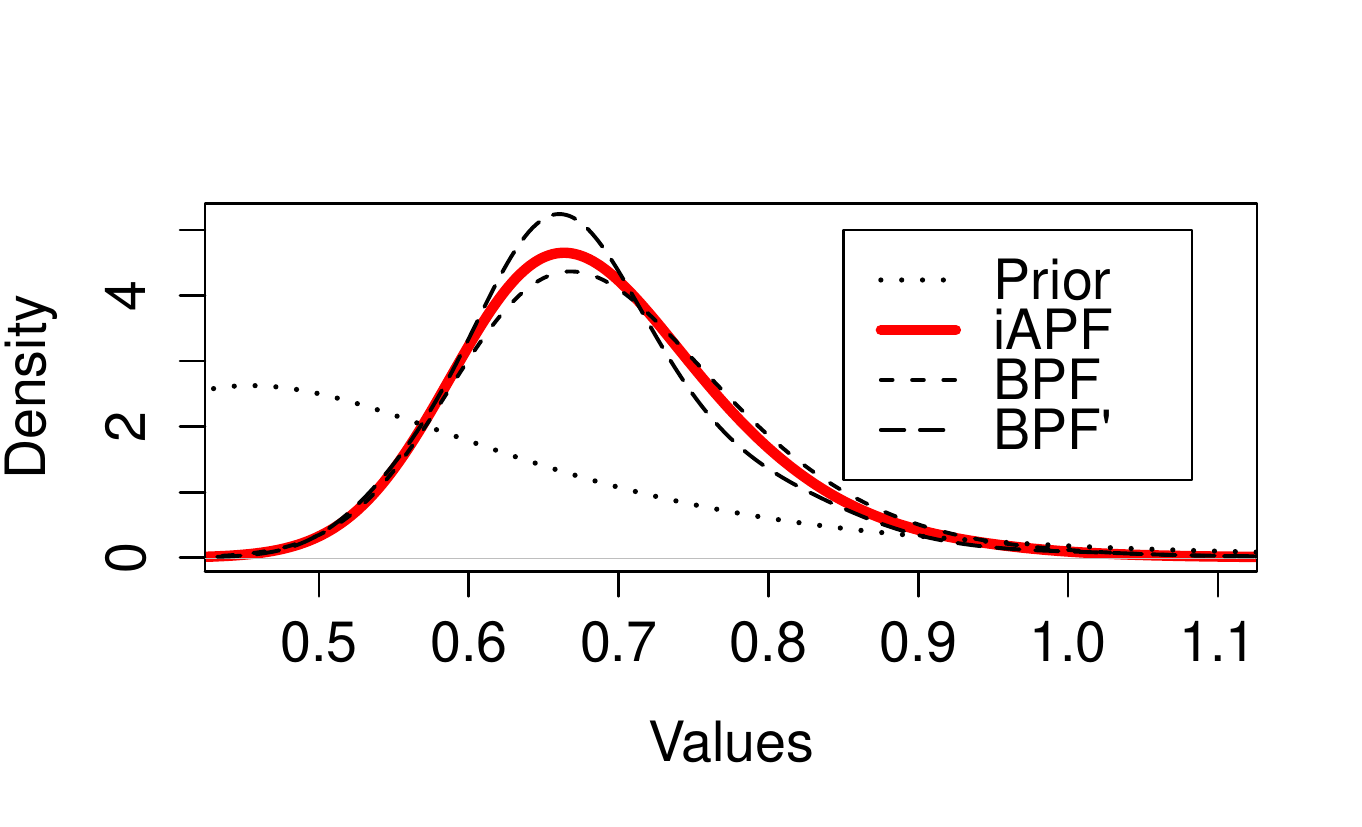}}\caption{\label{fig:Stochastic-Volatility-model:}Stochastic Volatility model:
PMMH density estimates for each parameter from the three chains.}
\end{figure}

In Table~\ref{tab:Sample-size-adjusted} we provide the adjusted
sample size of the Markov chains associated with each of the parameters,
obtained by dividing the length of the chain by the estimated integrated
autocorrelation time associated with each parameter. We can see an
improvement using the iAPF, although we note that the BPF-PMMH algorithm
appears to be fairly robust to the variability of the marginal likelihood
estimates in this particular application.

\begin{table}[H]
\noindent \centering{}\caption{\label{tab:Sample-size-adjusted}Sample size adjusted for autocorrelation
for each parameter from the three chains.}
\begin{tabular}{|c|c|c|c|}
\hline 
 & $\alpha$ & $\sigma^{2}$ & $\beta$\tabularnewline
\hline 
\hline 
iAPF & 3620 & 3952 & 3830\tabularnewline
\hline 
BPF & 2460 & 2260 & 3271\tabularnewline
\hline 
BPF' & 2470 & 2545 & 2871\tabularnewline
\hline 
\end{tabular}
\end{table}

Since particle filters provide approximations of the marginal likelihood
in HMMs, the iAPF can also be used in alternative parameter estimation
procedures, such as simulated maximum likelihood \citep{lerman1981use,diggle1984monte}.
The use of particle filters for approximate maximum likelihood estimation
\citep[see, e.g.,][]{kitagawa1998self,hurzeler2001approximating}
has recently been used to fit macroeconomic models \citep{fernandez2007estimating}.
In Figure~\ref{fig:Stochastic-Volatility-model:-2} we show the variability
of the BPF and iAPF estimates of the marginal likelihood at points
in a neighborhood of the approximate MLE of $(\alpha,\sigma,\beta)=(0.984,0.145,0.69)$.
The iAPF with $N_{0}=100$ particles used $100$ particles in the
final iteration to compute the likelihood in all simulations, and
took slightly more time than the BPF with $N=1000$ particles, but
far less time than the BPF with $N=10000$ particles. The results
indicate that the iAPF estimates are significantly less variable than
their BPF counterparts, and may therefore be more suitable in simulated
maximum likelihood approximations.

\begin{figure}[H]
\noindent \begin{centering}
\subfloat{\noindent \centering{}\includegraphics[width=15cm,height=7cm,keepaspectratio]{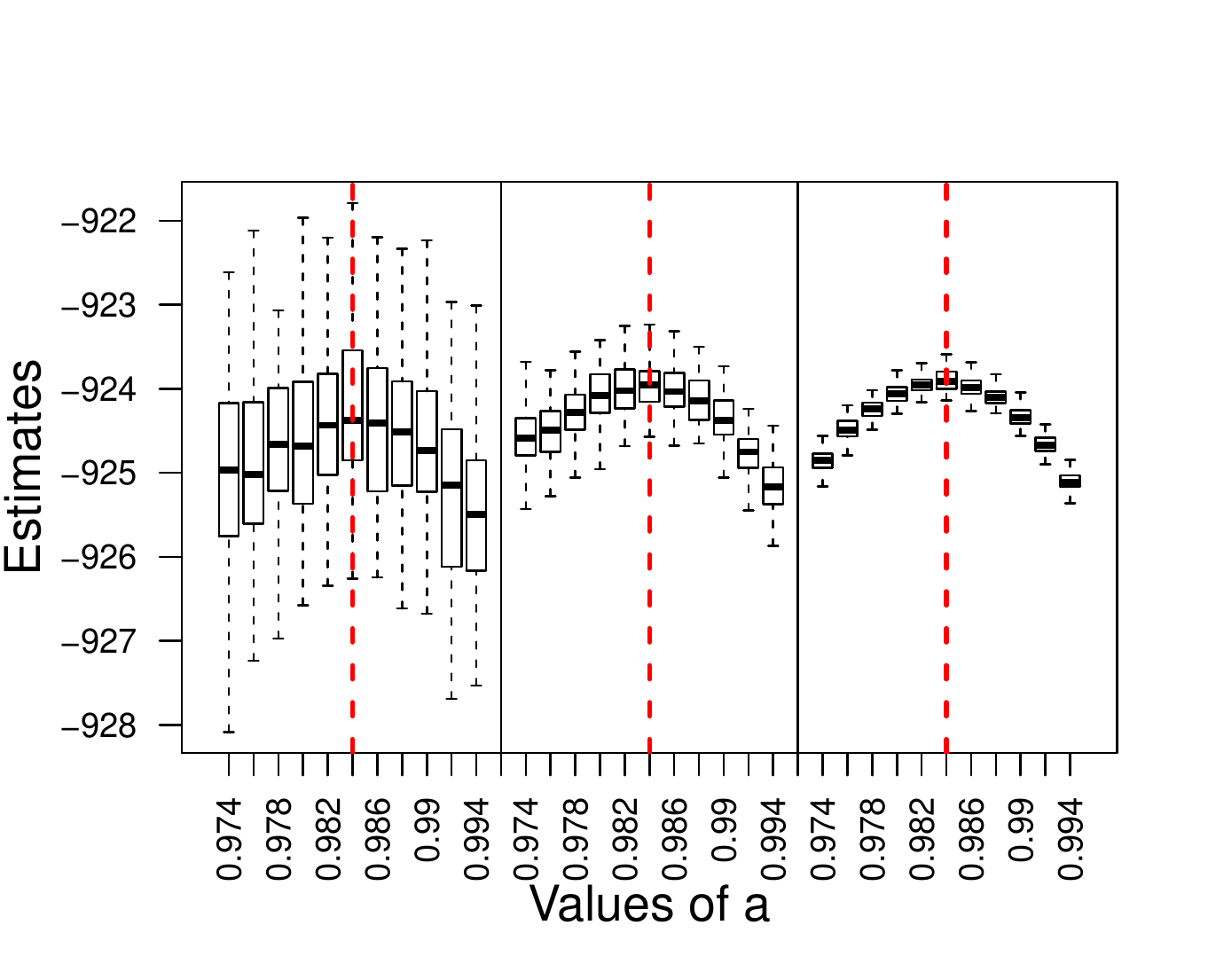}}\subfloat{\centering{}\includegraphics[width=15cm,height=7cm,keepaspectratio]{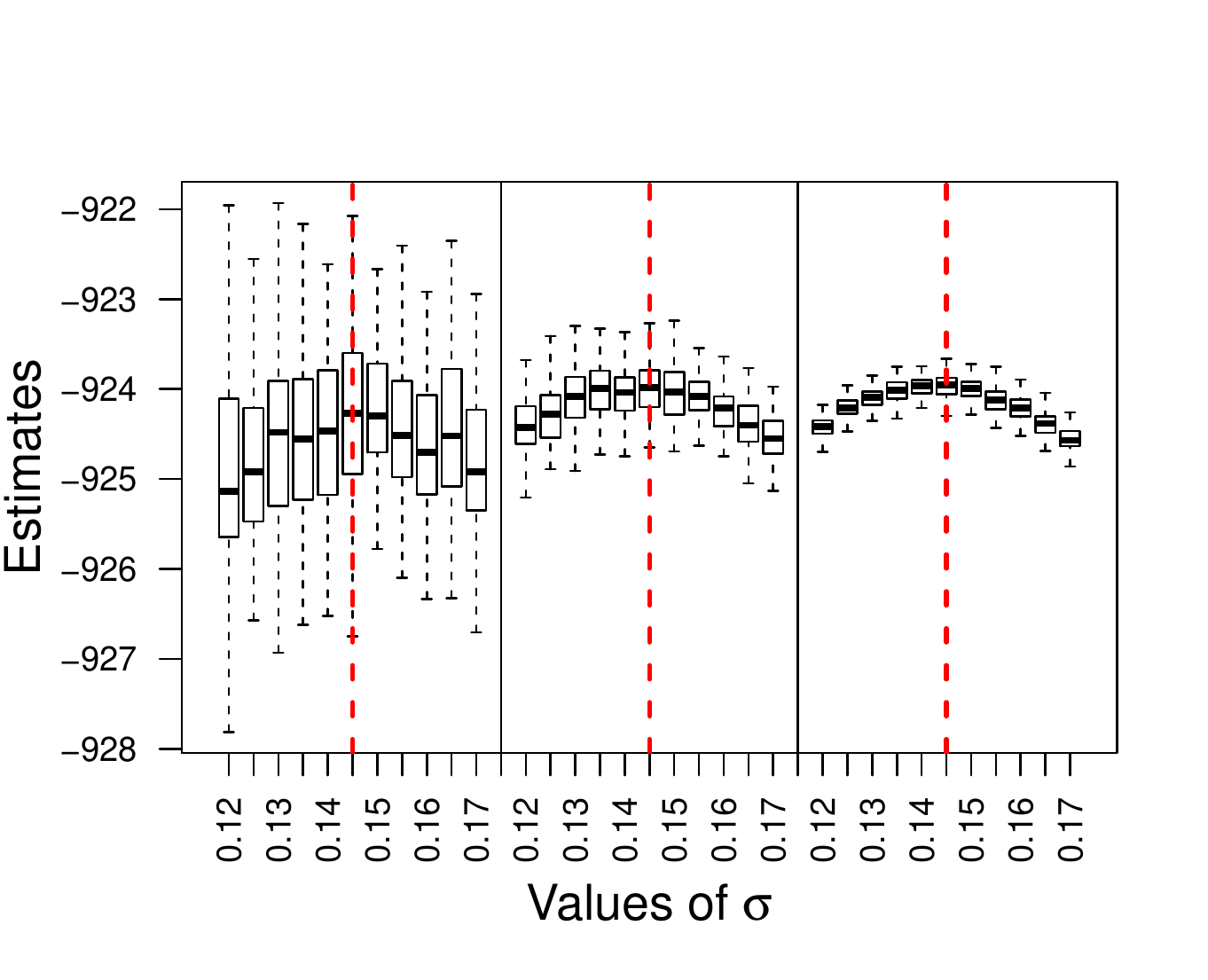}}
\par\end{centering}

\noindent \centering{}\subfloat{\centering{}\includegraphics[width=15cm,height=7cm,keepaspectratio]{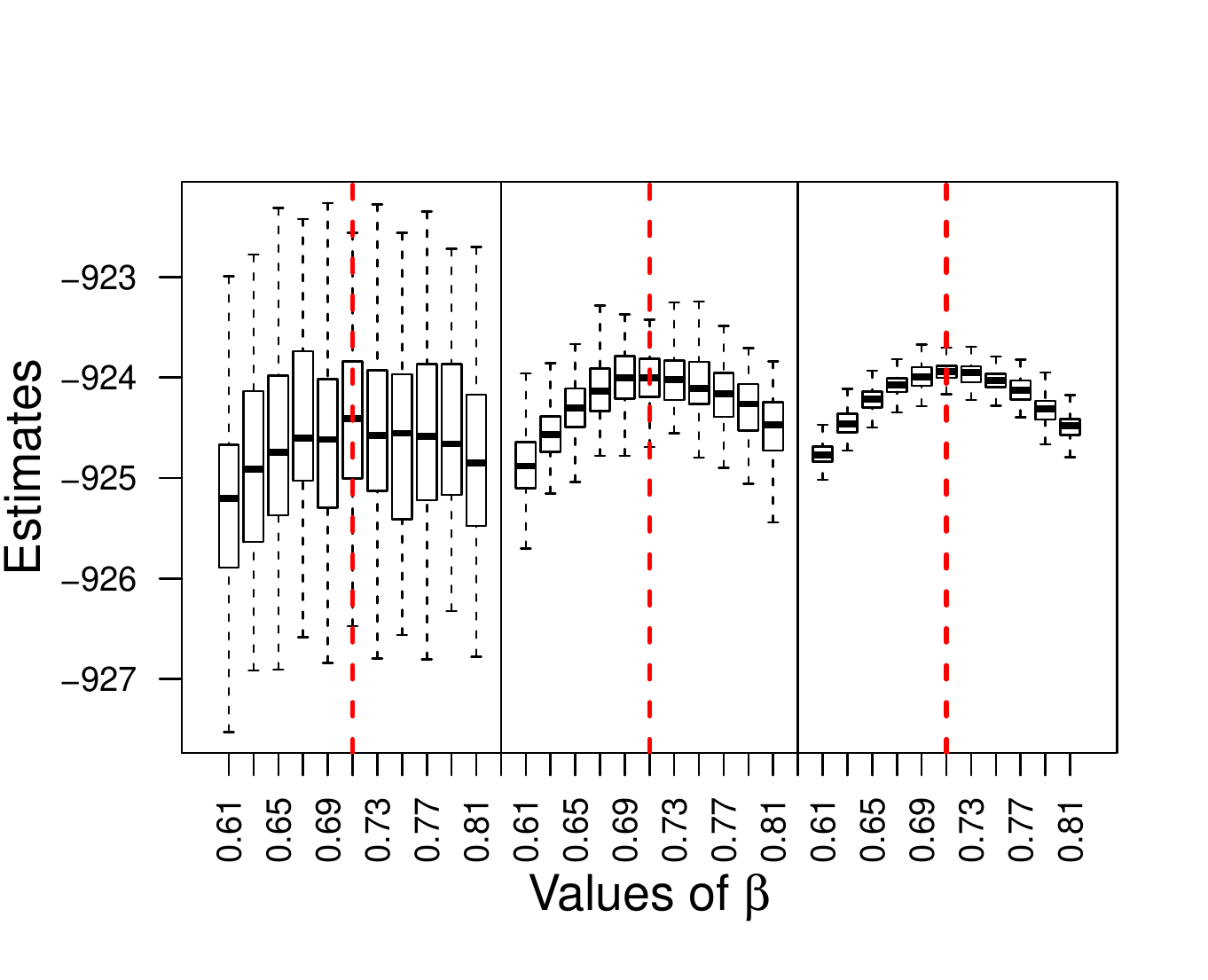}}\caption{\label{fig:Stochastic-Volatility-model:-2}log-likelihood estimates
in a neighborhood of the MLE. Boxplots correspond to $100$ estimates
at each parameter value given by three particle filters, from left
to right: BPF ($N=1000$), BPF ($N=10000$), iAPF ($N_{0}=100$).}
\end{figure}

\subsection{Multivariate stochastic volatility model\label{sub:Multivariate-Stochastic-Volatili}}

We consider a version of the multivariate stochastic volatility model
defined for $\mathsf{X}=\mathbb{R}^{d}$ by $\mu(\cdot)=\mathcal{N}(\cdot;m,U_{\star})$,
$f(x,\cdot)=\mathcal{N}(\cdot;m+{\rm diag}(\phi)\left(x-m\right),U)$
and $g(x,\cdot)=\mathcal{N}\left(\cdot;0,\exp\left({\rm diag}\left(x\right)\right)\right)$,
where $m,\phi\in\mbox{\ensuremath{\mathbb{R}}}^{d}$ and the covariance
matrix $U\in\mathbb{R}^{d\times d}$ are statistical parameters. The
matrix $U_{\star}$ is the stationary covariance matrix associated
with $(\phi,U)$. This is the \textit{basic MSV model} in \citet[Section~2]{chib2009multivariate},
with the exception that we consider a non diagonal transition covariance
matrix $U$ and a diagonal observation matrix.

We analyzed two 20-dimensional sequences of observations $y_{1:T}$
and $y_{1:T'}'$, where $T=102$ and $T'=90$. The sequences correspond
to the monthly returns for the exchange rate with respect to the US
dollar of a range of 20 different international currencies, in the
periods $3/2000$--$8/2008$ ($y_{1:T}$, pre-crisis) and $9/2008$--$2/2016$
($y_{1:T'}'$, post-crisis), as reported by the Federal Reserve System
(available at \url{http://www.federalreserve.gov/releases/h10/hist/}).
We infer the model parameters $\theta=\left(m,\phi,U\right)$ using
the iAPF to obtain marginal likelihood estimates within a PMMH algorithm.
A similar study using a different approach and with a set of $6$
currencies can be found in \citet{LiuWest2001}.

The aim of this study is to showcase the potential of the iAPF in
a scenario where, due to the relatively high dimensionality of the
state space, the BPF systematically fails to provide reasonable marginal
likelihood estimates in a feasible computational time. To reduce the
dimensionality of the parameter space we consider a band diagonal
covariance matrix $U$ with non-zero entries on the main, upper and
lower diagonals. We placed independent inverse Gamma prior distributions
with mean $0.2$ and unit variance on each entry of the diagonal of
$U$, and independent symmetric triangular prior distributions on
$[-1,1]$ on the correlation coefficients $\rho\in\mathbb{R}^{19}$
corresponding to the upper and lower diagonal entries. We place independent
Uniform$(0,1)$ prior distributions on each component of $\phi$ and
an improper, constant prior density for $m$. This results in a 79-dimensional
parameter space. As the starting point of the chains we used $\phi_{0}=0.95\cdot\mathbf{1}$,
${\rm diag}(U_{0})=0.2\cdot\mathbf{1}$ and for the $19$ correlation
coefficients we set $\rho_{0}=0.25\cdot\mathbf{1}$, where ${\bf 1}$
denotes a vector of $1$s whose length can be determined by context.
Each entry of $m_{0}$ corresponds to the logarithm of the standard
deviation of the observation sequence of the relative currency. 

We ran two Markov chains $X_{1:L}^{{\rm }}$ and $X'_{1:L}$, corresponding
to the data sequences $y_{1:T}$ and $y_{1:T'}'$, both of them updated
one component at a time with a Gaussian random walk proposal with
standard deviations $\left(0.2\cdot\mathbf{1},0.005\cdot\mathbf{1},0.02\cdot\mbox{\textbf{1}},0.02\cdot\mathbf{1}\right)$
for the parameters $\left(m,\phi,{\rm diag}\left(U\right),\rho\right)$.
The total number of updates for each parameter is $L=12000$ and the
iAPF with $N_{0}=500$ starting particles is used to estimate marginal
likelihoods within the PMMH algorithm. In Figure~\ref{fig:MSV-model:2}
we report the estimated smoothed posterior densities corresponding
to the parameters for the Pound Sterling/US Dollar exchange rate series.
Most of the posterior densities are different from their respective
prior densities, and we also observe qualitative differences between
the pre and post crisis regimes. For the same parameters, sample sizes
adjusted for autocorrelation are reported in Table~\ref{tab:Sample-size-adjusted-2}.
Considering the high dimensional state and parameter spaces, these
are satisfactory. In the later steps of the PMMH chain, we recorded
an average number of iterations for the iAPF of around $5$ and an
average number of particles in the final $\bm{\psi}$-APF of around
$502$.

\begin{table}[H]
\noindent \centering{}\caption{\label{tab:Sample-size-adjusted-2}Sample size adjusted for autocorrelation.}
\begin{tabular}{|c|c|c|c|c|}
\hline 
 & $m_{\text{£}}$ & $\phi_{\text{£}}$ & $U_{\text{£}}$ & $U_{\text{£,\officialeuro}}$\tabularnewline
\hline 
\hline 
pre-crisis & 408 & 112 & 218 & 116\tabularnewline
\hline 
post-crisis & 175 & 129 & 197 & 120\tabularnewline
\hline 
\end{tabular}
\end{table}
\begin{figure}[H]
\noindent \begin{centering}
\subfloat[$m_{\text{£}}$]{\noindent \centering{}\includegraphics[width=5cm,height=5cm]{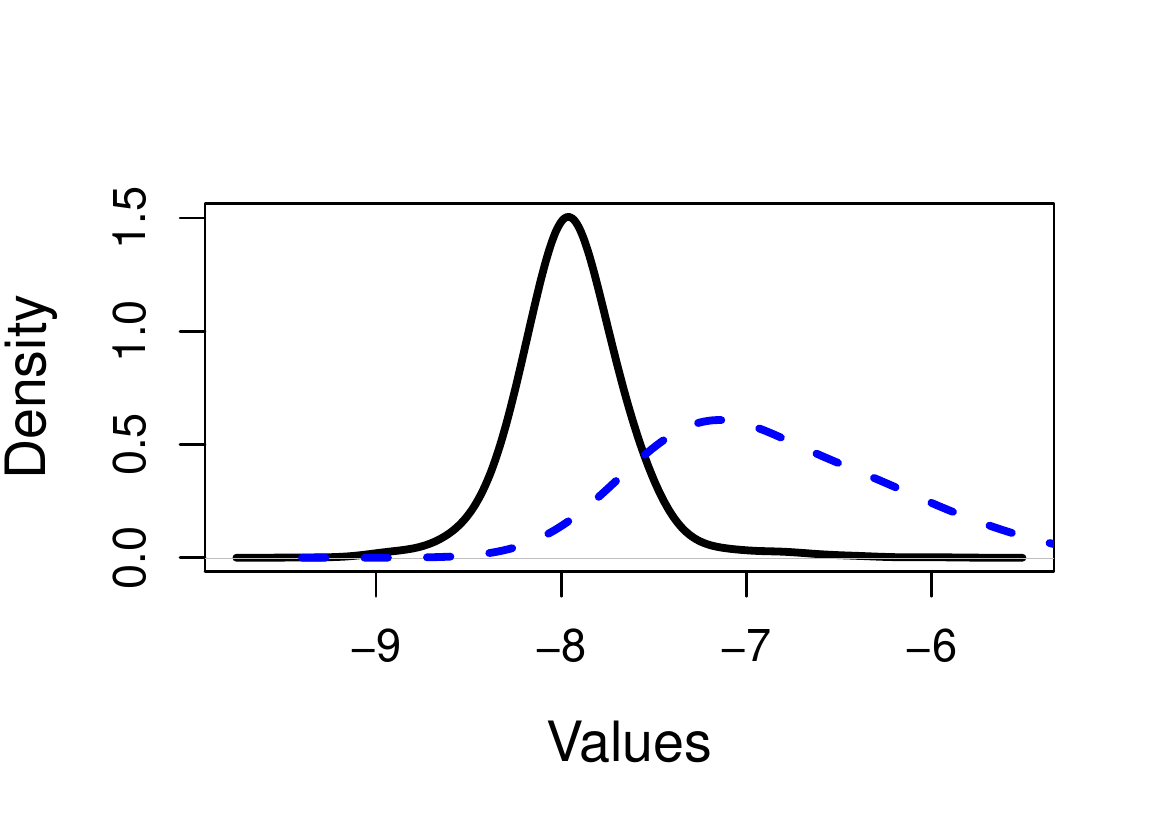}}\subfloat[$\phi_{\text{£}}$]{\noindent \centering{}\includegraphics[width=5cm,height=5cm]{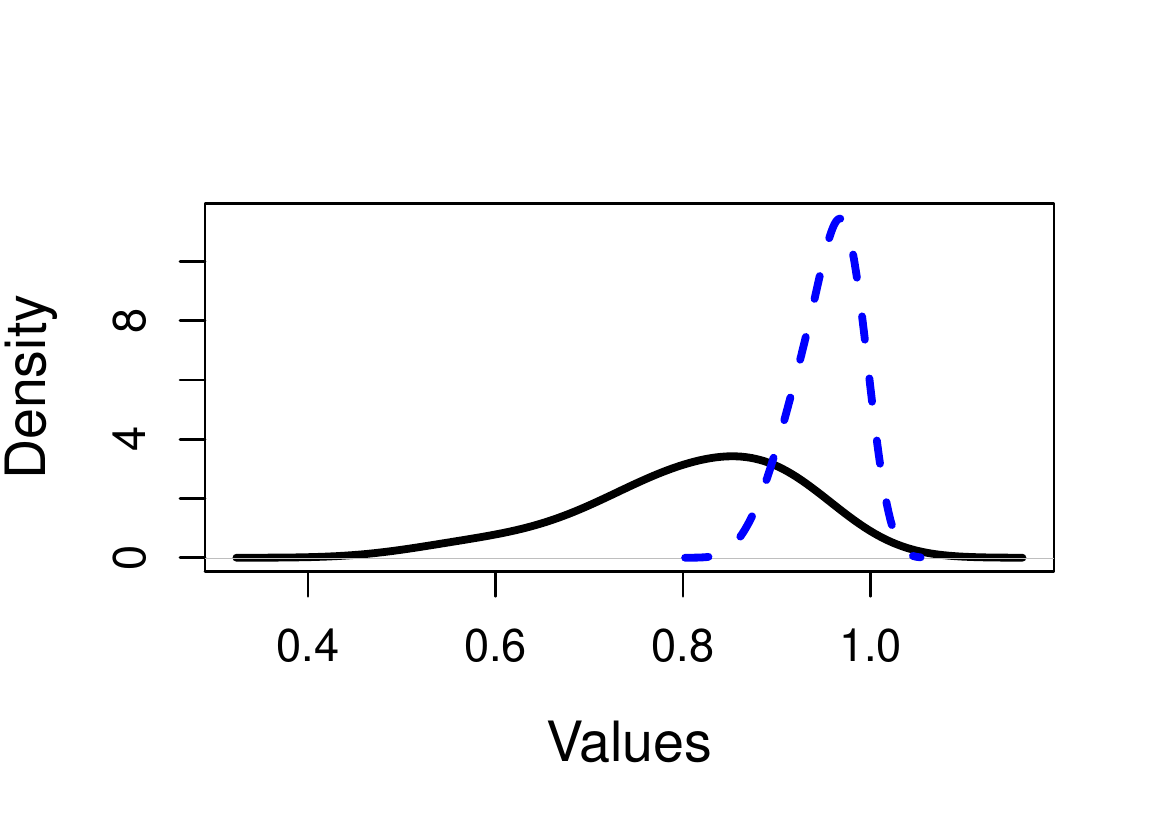}}
\par\end{centering}

\noindent \centering{}\subfloat[$U_{\text{£}}$]{\noindent \centering{}\includegraphics[width=5cm,height=5cm]{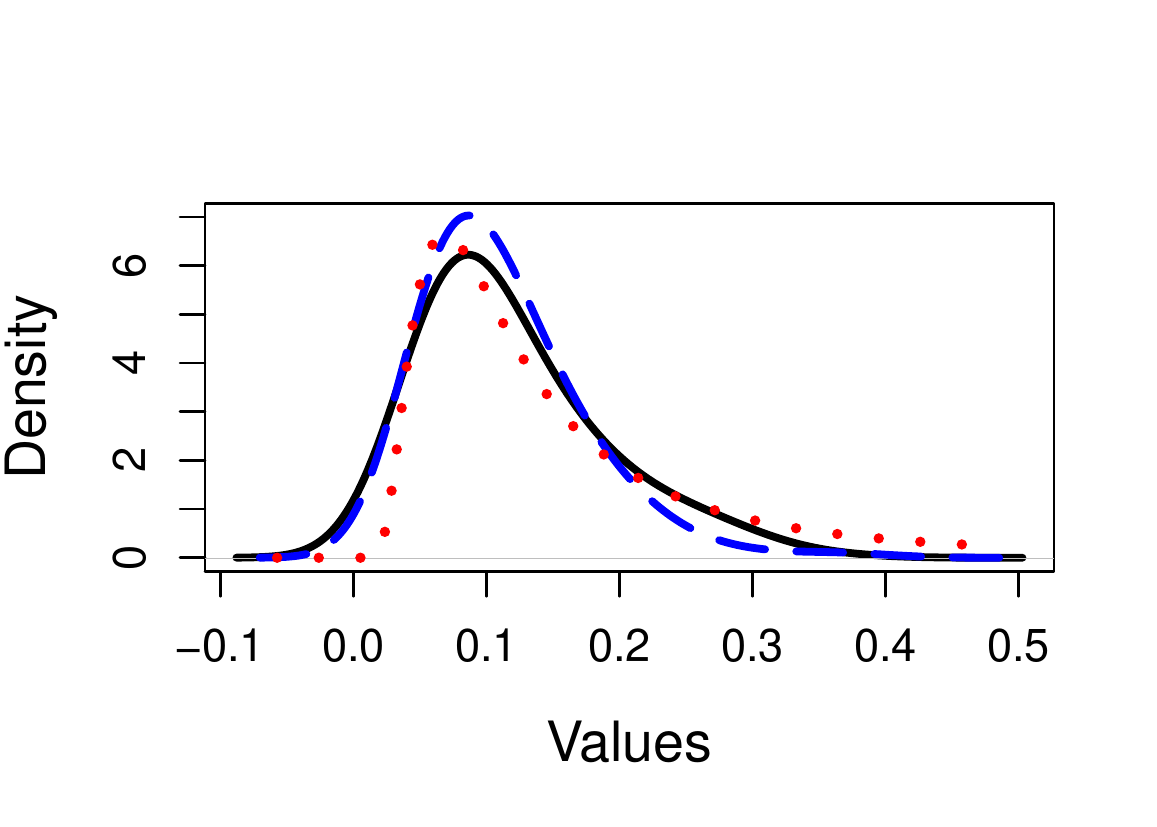}}\subfloat[$U_{\text{£,\officialeuro}}$]{\noindent \centering{}\includegraphics[width=5cm,height=5cm]{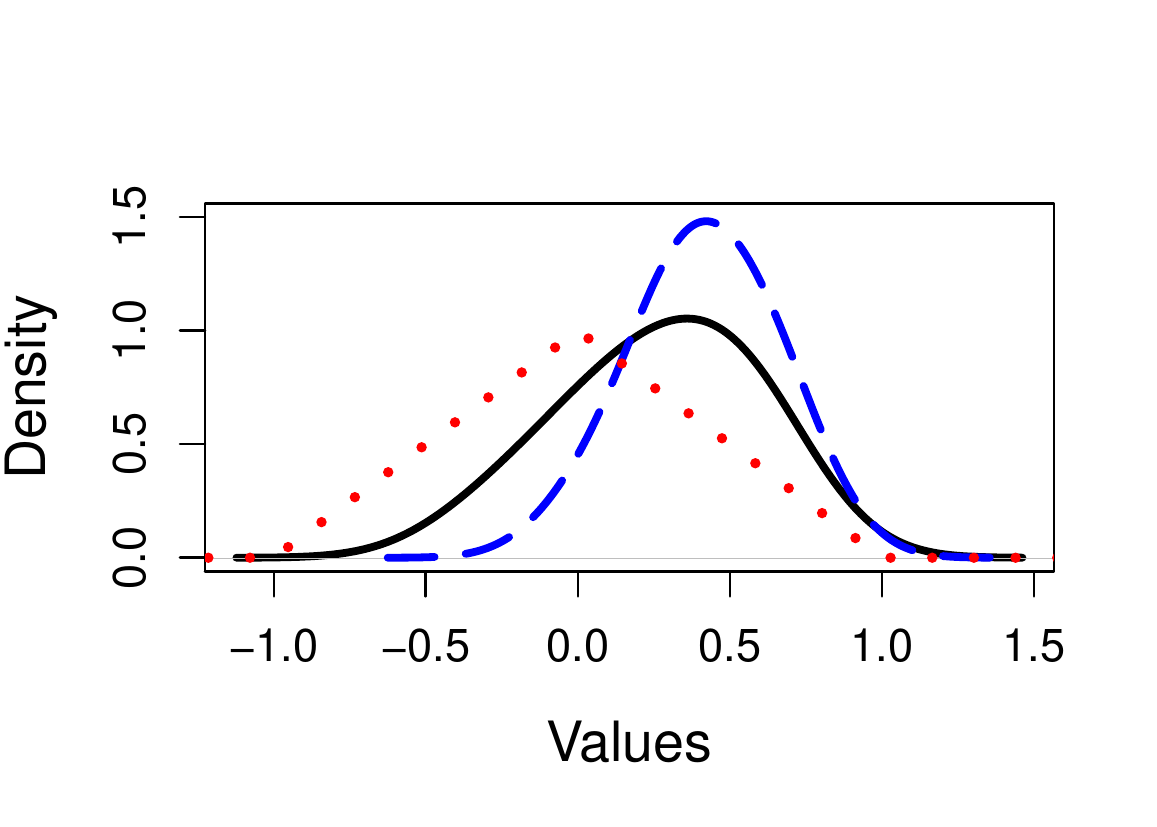}}\caption{\label{fig:MSV-model:2}Multivariate stochastic volatility model:
density estimates for the parameters related to the Pound Sterling.
Pre-crisis chain (solid line), post-crisis chain (dashed line) and
prior density (dotted line). The prior densities for (a) and (b) are
constant.}
\end{figure}

The aforementioned qualitative change of regime seems to be evident
looking at the difference between the posterior expectations of the
parameter $m$ for the post-crisis and the pre-crisis chain, reported
in Figure~\ref{fig:MSV-model:3}. The parameter $m$ can be interpreted
as the period average of the mean-reverting latent process of the
log-volatilities for the exchange rate series. Positive values of
the differences for close to all of the currencies suggest a generally
higher volatility during the post-crisis period.

\begin{figure}[H]
\noindent \centering{}\includegraphics{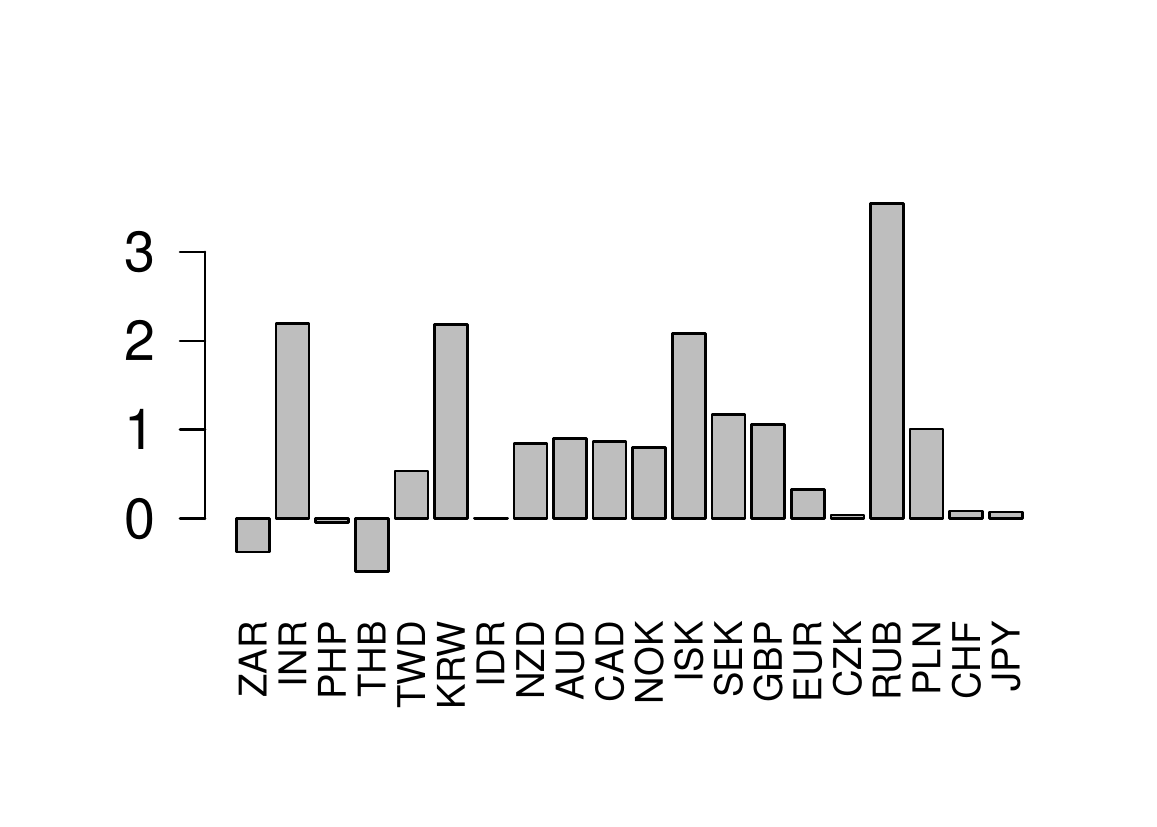}\caption{\label{fig:MSV-model:3}Multivariate stochastic volatility model:
differences between post-crisis and pre-crisis posterior expectation
of the parameter $m$ for the 20 currencies.}
\end{figure}

\section{Discussion\label{sec:Discussion}}

In this article we have presented the iAPF, an offline algorithm that
approximates an idealized particle filter whose marginal likelihood
estimates have zero variance. The main idea is to iteratively approximate
a particular sequence of functions, and an empirical study with an
implementation using parametric optimization for models with Gaussian
transitions showed reasonable performance in some regimes for which
the BPF was not able to provide adequate approximations. We applied
the iAPF to Bayesian parameter estimation in general state space HMMs
by using it as an ingredient in a PMMH Markov chain. It could also
conceivably be used in similar, but inexact, noisy Markov chains;
\citet{medina2015stability} showed that control on the quality of
the marginal likelihood estimates can provide theoretical guarantees
on the behaviour of the noisy Markov chain. The performance of the
iAPF marginal likelihood estimates also suggests they may be useful
in simulated maximum likelihood procedures. In our empirical studies,
the number of particles used by the iAPF was orders of magnitude smaller
than would be required by the BPF for similar approximation accuracy,
which may be relevant for models in which space complexity is an issue.

In the context of likelihood estimation, the perspective brought by
viewing the design of particle filters as essentially a function approximation
problem has the potential to significantly improve the performance
of such methods in a variety of settings. There are, however, a number
of alternatives to the parametric optimization approach described
in Section~\ref{sub:Implementation-details}, and it would be of
particular future interest to investigate more sophisticated schemes
for estimating $\bm{\psi}^{*}$, i.e. specific implementations of
Algorithm~\ref{alg:Function-approximations}. We have used nonparametric
estimates of the sequence $\bm{\psi}^{*}$ with some success, but
the computational cost of the approach was much larger than the parametric
approach. Alternatives to the classes $\mathcal{F}$ and $\Psi$ described
in Section~\ref{sub:Classes-of-} could be obtained using other conjugate
families, \citep[see, e.g.,][]{vidoni1999exponential}. We also note
that although we restricted the matrix $\Sigma$ in (\ref{eq:param})
to be diagonal in our examples, the resulting iAPF marginal likelihood
estimators performed fairly well in some situations where the optimal
sequence $\bm{\psi}^{*}$ contained functions that could not be perfectly
approximated using any function in the corresponding class. Finally,
the stopping rule in the iAPF, described in Algorithm~\ref{alg:iAPF}
and which requires multiple independent marginal likelihood estimates,
could be replaced with a stopping rule based on the variance estimators
proposed in \citet{lee2015variance}. For simplicity, we have discussed
particle filters in which multinomial resampling is used; a variety
of other resampling strategies \citep[see][for a review]{Doucet2005}
can be used instead.

\appendix

\section{Expression for the asymptotic variance in the CLT}
\begin{proof}[Proof of Proposition~\ref{prop:clt}]
We define a sequence of densities by 
\[
\pi_{k}^{\bm{\psi}}(x_{1:T}):=\frac{\left[\mu_{1}^{\bm{\psi}}\left(x_{1}\right)\prod_{t=2}^{T}f_{t}^{\bm{\psi}}\left(x_{t-1},x_{t}\right)\right]\prod_{t=1}^{k}g_{t}^{\bm{\psi}}\left(x_{t}\right)}{\int_{\mathsf{X}^{T}}\left[\mu_{1}^{\bm{\psi}}\left(x_{1}\right)\prod_{t=2}^{T}f_{t}^{\bm{\psi}}\left(x_{t-1},x_{t}\right)\right]\prod_{t=1}^{k}g_{t}^{\bm{\psi}}\left(x_{t}\right)dx_{1:T}},\quad x_{1:T}\in\mathsf{X}^{T},
\]
for each $k\in\{1,\ldots,T\}$. We also define $\pi_{k}^{\bm{\psi}}(x_{j}):=\int\pi_{k}(x_{1:j-1},x_{j},x_{j+1:T})dx_{-j}$
for $j\in\{1,\ldots,T\}$, where $x_{-j}:=(x_{1},\ldots,x_{j-1},x_{j+1},\ldots,x_{N})$.
Combining equation $(24.37)$ of \citet{Doucet2009} with elementary
manipulations provides, 
\begin{eqnarray*}
\sigma_{\bm{\psi}}^{2} & = & \sum_{t=1}^{T}\left[\int_{\mathsf{X}}\frac{\pi_{T}^{\bm{\psi}}(x_{t})^{2}}{\pi_{t-1}^{\bm{\psi}}(x_{t})}dx_{t}-1\right]\\
 & = & \sum_{t=1}^{T}\left[\int_{\mathsf{X}}\frac{\psi_{t}^{*}(x_{t})}{\psi_{t}(x_{t})}\pi_{T}^{\bm{\psi}}(x_{t})dx_{t}\cdot\frac{\int_{\mathsf{X}}\psi_{t}\left(x_{t}\right)\pi_{t-1}^{\bm{\psi}}(x_{t})dx_{t}}{\int_{\mathsf{X}}\psi_{t}^{*}\left(x_{t}\right)\pi_{t-1}^{\bm{\psi}}(x_{t})dx_{t}}-1\right]\\
 & = & \sum_{t=1}^{T}\left\{ \mathbb{E}\left[\frac{\psi_{t}^{*}\left(X_{t}\right)}{\psi_{t}\left(X_{t}\right)}\Bigl|\left\{ Y_{1:T}=y_{1:T}\right\} \right]\frac{\mathbb{E}\left[\psi_{t}\left(X_{t}\right)\mid\left\{ Y_{1:t-1}=y_{1:t-1}\right\} \right]}{\mathbb{E}\left[\psi_{t}^{*}\left(X_{t}\right)\mid\left\{ Y_{1:t-1}=y_{1:t-1}\right\} \right]}-1\right\} ,
\end{eqnarray*}
and the expression involving the rescaled terms $\bar{\psi}_{t}^{*}$
and $\bar{\psi}_{t}$ then follows.
\end{proof}
\bibliographystyle{agsm}
\bibliography{Biblio}

\end{document}